\documentclass[journal]{IEEEtran}

\usepackage{graphicx}
\usepackage{color}
\usepackage{placeins}
\usepackage{float}
\usepackage{tabularx,colortbl}
\usepackage{amssymb}
\usepackage{amsthm}
\usepackage{algorithm}  
\usepackage{algorithmic}
\usepackage{cite}
\usepackage{url}
\usepackage{amsmath}
\usepackage{makecell}
\usepackage{subfigure}
\usepackage{caption}
\usepackage{enumitem}
\usepackage{booktabs}

\theoremstyle{plain}

\newtheorem{lemma}{Lemma}
\newtheorem{proposition}{Proposition}
\theoremstyle{plain}

\IEEEoverridecommandlockouts

\begin{document}
	\title{Fundamental CRB-Rate Tradeoff of Rydberg Atomic Receivers for Integrated Sensing and Communication}
	\author{Yin Zhang, Jiayi Zhang,~\IEEEmembership{Senior Member,~IEEE}, Bokai Xu,\\ Jieao~Zhu,~\IEEEmembership{Student Member,~IEEE}, and Bo Ai,~\IEEEmembership{Fellow,~IEEE}
		\thanks{Y. Zhang, J. Zhang, B. Xu, and B. Ai are with the State Key Laboratory of Advanced Rail Autonomous Operation and the School of Electronics and Information Engineering, Beijing Jiaotong University, Beijing 100044, China (e-mails: 25120225@bjtu.edu.cn, zhangjiayi@bjtu.edu.cn, 20251197@bjtu.edu.cn, boai@bjtu.edu.cn).
		J. Zhu is with the Huawei Wireless RAN Research Department, Shanghai, China (e-mail: zja17@tsinghua.org.cn). }
	}
	\IEEEaftertitletext{\vspace{-2\baselineskip}}
	\maketitle
	\begin{abstract}
		Integrated sensing and communication (ISAC) is a key technology for next generation wireless networks, requiring highly sensitive, compact receiver front ends with multiband coverage.
		Rydberg atomic receivers (RARs) provide a promising architecture that satisfies these requirements owing to their high sensitivity and continuous frequency tuning capability.
		Unlike in classical ISAC systems utilizing orthogonal resource allocation, executing concurrent multi band tasks within a single atomic ensemble relies on shared energy state transitions.
		Owing to this shared structure, the communication and sensing tradeoff in RAR-based ISAC introduces an atomic-response tradeoff beyond the classical subspace and deterministic-random tradeoffs. This paper studies a multi-branch RAR (MB-RAR) architecture for joint communication and radar sensing. By solving the five level Lindblad master equation at the local oscillator (LO) biased operating point and linearizing the communication and sensing radio frequency (RF) signals as small perturbations, we obtain the double-sideband dynamic signal response of the MB-RAR. 
The model reveals two intrinsic constraints. First, increasing the coupling laser strength broadens the response bandwidth but can weaken the atomic response. Second, strengthening one branch's LO can suppress the other branch's conversion gain through the shared Rydberg state, even when the latter's LO bias remains fixed, which is termed the inter band gain inhibition (IBGI). Based on these mechanisms, we formulate the achievable Cramér--Rao bound (CRB)-rate region and develop a rate-constrained joint search over coupling strength and LO power allocation.
		Numerical results validate the proposed response model and show that joint optimization provides a more complete characterization of the achievable CRB-rate region. Compared with idealized RF-ISAC time sharing, the RAR achieves a lower CRB endpoint, whereas the RF-ISAC attains a higher communication-rate endpoint.
	\end{abstract}
	
	\section{Introduction}
	The evolution toward sixth generation (6G) wireless networks is predicated on the convergence of the physical and digital worlds, with integrated sensing and communication (ISAC) emerging as a critical enabler \cite{gong2026rydbergTC, liu2025sensing}. By unifying radar sensing and data transmission into a single hardware platform, ISAC enables the joint reuse of spectrum and hardware resources \cite{xiong2024torch}. However, current ISAC implementations, predominantly based on classical solid state electronics, are approaching physical bottlenecks \cite{xu2023resource}. These bottlenecks originate from the dependence of existing systems on shared radio frequency (RF) front end components, such as metallic antennas, low noise amplifiers (LNAs), and mixers. In addition, the physical dimensions of classical metallic antennas strictly dictate their resonance frequencies and, as captured by the Chu limit, fundamentally bound their instantaneous bandwidth and efficiency \cite{hua2023mimo}. Their conductive structures also introduce metallic scattering and mutual coupling, which can aggravate self interference in dense arrays \cite{zeng2025holographic}. These limitations motivate alternative reception mechanisms beyond conventional electron charge induction, including physics based quantum reception \cite{gong2026rydbergTC, yuan2025electromagnetic}.
	
Rydberg atomic receivers (RARs) provide an alternative RF reception mechanism based on quantum sensing \cite{chen2025harnessing, xu2025channel}. In a typical RAR, the incident RF field is mapped to an optical response of the atomic medium through Rydberg state transitions, while a strong microwave local oscillator (LO) enables quantum heterodyne reception and downconverts weak RF signals into measurable optical intensity variations \cite{EIT, jing2020atomic}. This reception principle provides several features that are relevant to ISAC systems. First, the resonant frequencies of RARs are determined by atomic energy level transitions rather than by metallic antenna geometries, enabling electric field sensing over a broad spectral range, from kHz bands \cite{zhang2026self} to THz bands \cite{yao2026ultra}. Second, the all dielectric vapor cell structure has a small radar cross section (RCS) and can reduce metallic scattering and antenna mutual coupling. Third, optical readout provides electrical isolation and avoids part of the cascaded thermal noise accumulation introduced by conventional RF front end chains. These properties make RARs attractive for ISAC, where communication and sensing functions require sensitive, compact, and reconfigurable RF reception \cite{cui2024multi,gong2025rydbergMU}. In particular, multiple microwave transitions can be addressed within the same atomic ensemble, enabling concurrent multi band reception without separate RF front ends or combiners for each band \cite{cui2025towards, QS-34}.
	
Existing research has demonstrated the feasibility of RARs for individual functionalities, including high fidelity AM, FM, phase modulated, and QAM communication demodulation \cite{song2019rydberg, anderson2020atomic, holloway2019detecting} as well as high sensitivity electric field metrology and radar sensing \cite{QS-24, gordon2019weak, continu4}. Recent multi band and multi carrier RAR studies have further shown that multiple microwave channels can be addressed within a common atomic medium \cite{zhu2025raq, wang2025multi, kim2025multi}. However, the fundamental communication and sensing tradeoff of such receivers remains insufficiently understood. In classical ISAC systems, this tradeoff is typically characterized through the subspace tradeoff and deterministic random tradeoff, which depend on resource allocation and signaling design, respectively \cite{xiong2024torch, xiong2023fundamental}. In RAR-based ISAC, these classical tradeoffs are further coupled with the atomic response of the receiver. When communication and sensing signals are simultaneously received by a single atomic ensemble, the two functions are mediated by shared Rydberg state coherences and microwave induced transitions. Consequently, the achievable Cramér--Rao bound (CRB)-rate (C-R) performance is determined not only by waveform or resource allocation, but also by how the coupling laser and LO branches configure the common atomic response. A system level model is therefore needed to quantify the bandwidth response tradeoff, inter band gain inhibition (IBGI), and their impact on the C-R region of RAR-based ISAC.
	
	To address this gap, this paper studies a multi-branch RAR (MB-RAR) architecture for ISAC and develops an analytical framework for characterizing its C-R tradeoff. The MB-RAR architecture uses different LO branches to address resolvable Rydberg microwave transitions, while all branches are still coupled through the shared Rydberg state and the common optical electromagnetically induced transparency (EIT) channel. This structure makes the RAR-ISAC tradeoff different from that of conventional RF-ISAC systems, because the communication and sensing performance is shaped not only by signaling or resource allocation, but also by the receiver side atomic response. The main contributions are summarized as follows.
	\begin{itemize}
		\item \textbf{Unified Quantum Transconductance Modeling:} We derive the dynamic signal response of the MB-RAR by solving the five level Lindblad master equation at the LO biased operating point and linearizing the communication and sensing RF signals as small perturbations. The resulting branch atomic response connects the microscopic density matrix dynamics with the equivalent communication and sensing observation models.
		
		\item \textbf{Analytical Characterization of Intrinsic Physical Tradeoffs:} We characterize two RAR specific effects that govern the C-R tradeoff. First, the coupling laser Rabi frequency $\Omega_c$ controls the common EIT channel. It changes the Autler--Townes (AT) splitting and the dynamic small signal response bandwidth, which can support a higher communication rate over an appropriate operating range, but it also changes the atomic response strength.
Second, the LO branches jointly determine the atomic populations, coherences, and optical conversion response. We quantify IBGI by comparing branch gains at the same own-branch LO bias, and separate the allocation sensitivity into direct mixing, own-branch response, and cross-branch response terms. Numerical validation shows the impact of IBGI and its optimal coupling settings.
		
		\item \textbf{Characterization of the Quantum C-R Region:} Based on the derived branch atomic response, we formulate the achievable C-R region of the MB-RAR ISAC. We further implement a rate-constrained joint optimization (RCJO) algorithm based on the $\epsilon$-constraint formulation and feasible coordinate updates to characterize the Pareto boundary. Simulation results show that joint tuning balances bandwidth broadening and IBGI. Additionally, compared with conventional RF-ISAC time sharing, RAR-based ISAC achieves smaller sensing CRBs under the default parameters, while the RF-ISAC baseline attains a higher communication-rate endpoint through its larger receive signal-to-noise ratio (SNR).
	\end{itemize}   
	
	The remainder of this paper is organized as follows. Section \ref{sec:isac_sysmodel} describes the MB-RAR system architecture, the underlying quantum master equation, and the Liouvillian small signal response model. Section \ref{sec:tradeoffs} discusses the physical nature of the EIT window and LO power tradeoffs. Section \ref{sec:C-R} formulates the communication rate and sensing CRB based on the quantum transconductance, characterizes the achievable C-R region, and details the proposed RCJO algorithm. Numerical results and performance analyses are presented in Section \ref{sec:numerical}, and Section \ref{sec:conclusion} concludes the paper.
	
	\textit{Notation:} Throughout this paper, boldface uppercase letters (e.g., $\mathbf{H}$) denote matrices, and boldface lowercase letters (e.g., $\mathbf{a}$) denote vectors. The superscripts $(\cdot)^{\mathrm{T}}$, $(\cdot)^{*}$, and $(\cdot)^{\mathrm{H}}$ denote the transpose, complex conjugate, and conjugate transpose, respectively. $\|\cdot\|$ denotes the Euclidean norm, $\mathbf I_M$ denotes the $M\times M$ identity matrix, $\operatorname{Re}\{\cdot\}$ and $\operatorname{Im}\{\cdot\}$ denote the real and imaginary parts, $\mathrm{vec}(\cdot)$ denotes column-wise vectorization, and $\mathrm{Tr}(\cdot)$ denotes the trace. $\mathbb{E}[\cdot]$ denotes the statistical expectation, and $\mathcal{CN}(\mu, \sigma^2)$ represents a circularly symmetric complex Gaussian distribution. The symbol $\hbar$ denotes the reduced Planck constant.
	
	\section{ISAC System Modeling}
	\label{sec:isac_sysmodel}
	\begin{figure}[!t]
		\centering
		\includegraphics[width=3in]{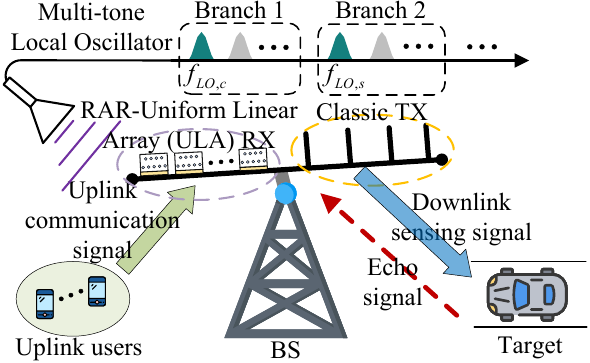}
		\caption{The architecture of the multi-branch Rydberg atomic receivers for integrated sensing and communication, where the communication and sensing LO branches are explicitly denoted by $\mathrm{LO}_{c}$ and $\mathrm{LO}_{s}$, respectively.}
		\label{fig:system_architecture}
	\end{figure}
	We consider an uplink ISAC system equipped with an MB-RAR array. As illustrated in Fig.~\ref{fig:system_architecture}, the receiver adopts the principle of heterodyne reception and is driven by a multi tone LO field \cite{zhu2025raq,wang2025multi,kim2025multi,cui2024multi,gong2025rydbergMU}. The MB-RAR architecture is organized into task oriented atomic response branches, where each branch is associated with one ISAC function.
	
	\subsection{Radio Frequency Signal Model}
	The MB-RAR array contains $M$ independent vapor cells, where the output of the $m$-th vapor cell forms the $m$-th array observation. The vapor cells are arranged along the uniform linear array (ULA) direction, and adjacent cells are separated by the inter-element spacing $d$. Each vapor cell has a length $\ell_{\rm cell}$ along its optical axis, which is perpendicular to the ULA direction, as shown in Fig.~\ref{fig:system_architecture}.
	
	Under the peak complex-envelope convention, a plane wave with power $P$ and electric-field amplitude $E$ over the effective receiver aperture $A_{\mathrm{e}}$ satisfies $P=A_{\mathrm{e}}|E|^2/(2\eta_0)$, where $\eta_0 = 377\ \Omega$ is the free space intrinsic impedance. Therefore, the incident communication electric field at the $m$-th array element is
	\begin{equation}
		E_{\mathrm{com}}^{(m)}(t)
		=
		\sqrt{\frac{2\eta_0P_{\mathrm{tx}}}{A_{\mathrm{e}}}}
		[\mathbf{h}_{\mathrm{com}}]_m x_{\mathrm{com}}(t)
		\exp\left(j 2\pi f_c t\right),
		\label{eq:field_com}
	\end{equation}
	where $P_{\mathrm{tx}}$ is the transmit power, $x_{\mathrm{com}}(t)$ is the normalized complex baseband symbol stream satisfying $\mathbb{E}[|x_{\mathrm{com}}(t)|^2]=1$, and $\mathbf{h}_{\mathrm{com}}\in\mathbb{C}^{M\times 1}$ is the dimensionless RF channel including the angle-dependent ULA phase. The effective aperture is $A_{\mathrm{e}}=2\pi w_0^2+2\pi w_0\ell_{\rm cell}$, where $w_0$ is the optical beam waist and $\ell_{\rm cell}$ is the vapor-cell length along the optical axis \cite{gong2026rydbergTC}.
	
	For the sensing function, the base station (BS) transmits a probing waveform toward a static point target located at angle $\theta$, so that the Doppler shift due to target motion is neglected in the present angle-of-arrival (AoA) estimation model. The reflected echo arrives at the receiver array with round trip delay $\tau$ and carries the spatial phase signature required for AoA estimation. The sensing echo field at the $m$-th array element is expressed as~\cite{xiong2023fundamental, gong2025rydbergMT} 
	\begin{equation}
		E_{\mathrm{sens}}^{(m)}(t)
		=
		\sqrt{\frac{2\eta_0P_{\mathrm{echo}}}{A_{\mathrm{e}}}}
		x_s(t-\tau)\exp\left(j 2\pi f_s(t-\tau)\right)[\mathbf{a}(\theta)]_m,
		\label{eq:field_sens}
	\end{equation}
	where $P_{\mathrm{echo}}$ is the equivalent echo power incident on $A_{\mathrm{e}}$ at each array element, $x_s(t)$ is the normalized probing waveform, and $f_s$ is the sensing carrier frequency.
	The sensing spatial vector is given by $[\mathbf{a}(\theta)]_m=a_m(\theta)=e^{j(m-1)kd\sin\theta}$, where $k=2\pi/\lambda_s$, $\lambda_s$ is the sensing wavelength, and $d$ is the inter-element spacing.
	
	To facilitate frequency down conversion, the atomic array is driven by a multi tone LO field generated by a common external reference. The LO field is assumed to uniformly illuminate all vapor cells with synchronized phases; hence, it is spatially invariant across array elements, i.e., $E_{\mathrm{LO}}^{(m)}(t)=E_{\mathrm{LO}}(t), \forall m$. Then, the combined LO field ${E}_{\mathrm{LO}}(t)$ can be modeled as
	\begin{equation}
		{E}_{\mathrm{LO}}(t) = {E}_{\mathrm{LO},c} \exp(j 2\pi f_{\mathrm{LO},c} t) + {E}_{\mathrm{LO},s} \exp(j 2\pi f_{\mathrm{LO},s} t),
		\label{eq:field_lo}
	\end{equation}
	where $E_{\mathrm{LO},c}$ and $E_{\mathrm{LO},s}$ correspond to the amplitudes of the communication and sensing LOs, respectively. The total RF field is ${E}_{\mathrm{total}}^{(m)}(t)=E_{\mathrm{LO}}^{(m)}(t)+E_{\mathrm{com}}^{(m)}(t)+E_{\mathrm{sens}}^{(m)}(t)$.
	
	\begin{figure}[!t]
		\centering
		\includegraphics[width=3.5in]{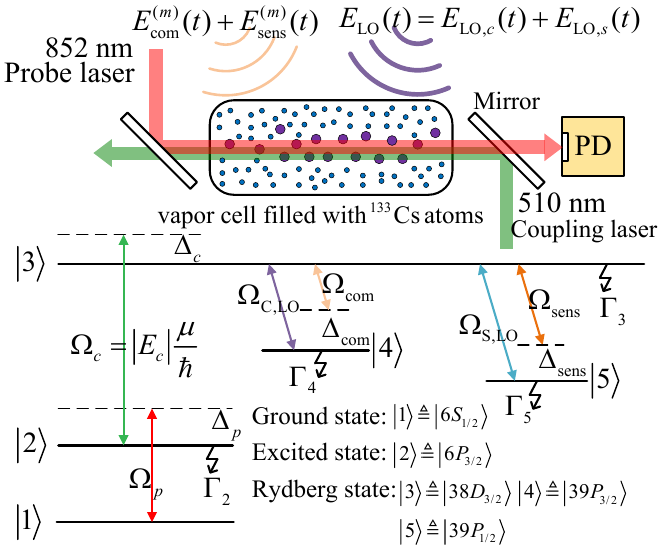}
		\caption{Illustration of the Rydberg atomic receiver structure and the energy level diagram of the Cs atom.}
		\label{fig:five_level}
	\end{figure}
	
	\subsection{Double-Sideband Quantum Response of MB-RAR}
	As shown in Fig.~\ref{fig:five_level}, we consider a cesium (Cs) vapor cell built on a five level ladder plus branch diagram \cite{simons2018fiber, zhu2025general}.
	The optical EIT backbone consists of $\left|1\right\rangle\triangleq\left|6S_{1/2}\right\rangle$, $\left|2\right\rangle\triangleq\left|6P_{3/2}\right\rangle$, and $\left|3\right\rangle\triangleq\left|38D_{3/2}\right\rangle$. The probe field with wavelength $\lambda_p=852$~nm drives $\left|1\right\rangle\leftrightarrow\left|2\right\rangle$ with Rabi frequency $\Omega_p$ and detuning $\Delta_p$, while the coupling field with wavelength $\lambda_c\approx510.5$~nm drives $\left|2\right\rangle\leftrightarrow\left|3\right\rangle$ with Rabi frequency $\Omega_c$ and detuning $\Delta_c$, establishing the EIT pathway.

	For Rydberg atoms in a room-temperature vapor, thermal atomic motion causes Doppler shifts of both the probe and coupling lasers, which modify the atomic density matrix. For the adopted counter-propagating optical geometry, let $v$ denote the atomic velocity component along the probe-beam direction. The corresponding optical detunings are
	\begin{equation}
	\Delta_p(v)=\Delta_p-k_pv,\qquad
	\Delta_c(v)=\Delta_c+k_cv,
	\label{eq:optical_doppler_detunings}
	\end{equation}
	where $k_p=2\pi/\lambda_p$ and $k_c=2\pi/\lambda_c$. The longitudinal velocity distribution is
	\begin{equation}
	f_T(v)=
	\frac{1}{\sqrt{2\pi}\sigma_v}
	\exp\left(-\frac{v^2}{2\sigma_v^2}\right),
	\qquad
	\sigma_v^2=\frac{k_BT}{m_{\mathrm{Cs}}},
	\label{eq:maxwell_boltzmann_velocity}
	\end{equation}
	where $k_B$ is the Boltzmann constant, $m_{\mathrm{Cs}}$ is the mass of a Cs atom, and $T$ is the operating temperature of the vapor cell.
	
	The branches of MB-RAR are formed by coupling the shared Rydberg state $\left|3\right\rangle$ to two resolvable upper Rydberg states, $\left|4\right\rangle\triangleq\left|39P_{3/2}\right\rangle$ and $\left|5\right\rangle\triangleq\left|39P_{1/2}\right\rangle$. The communication branch $\left|3\right\rangle\leftrightarrow\left|4\right\rangle$ has atomic transition frequency $f_{34}^{\mathrm{at}}=12.3576$~GHz, and the sensing branch $\left|3\right\rangle\leftrightarrow\left|5\right\rangle$ has atomic transition frequency $f_{35}^{\mathrm{at}}=17.1695$~GHz. The LO frequencies are tuned near the corresponding atomic transitions. The beat frequencies between the RF carriers and their corresponding LO tones, $f_{\mathrm{IF},\mathrm{com}}=|f_c-f_{\mathrm{LO},c}|$ and $f_{\mathrm{IF},\mathrm{sens}}=|f_s-f_{\mathrm{LO},s}|$, respectively, constitute the communication and sensing intermediate frequencies in atomic heterodyne reception.

For branch $b\in\{\mathrm{com},\mathrm{sens}\}$, the microwave Rabi frequency is decomposed into a strong LO bias and a weak signal perturbation,
\begin{equation}
\Omega_b(t)=\Omega_{b,\mathrm{LO}}+\delta\Omega_b(t),\quad
|\delta\Omega_b(t)|\ll|\Omega_{b,\mathrm{LO}}|.
\label{eq:branch_rabi_decomp}
\end{equation}

Throughout the atomic response model, $\delta$ denotes a small perturbation, while $\Delta_p$, $\Delta_c$, and $\Delta_b$ denote detunings. The two LO biases jointly establish one atomic operating point; $\delta\Omega_{\mathrm{com}}(t)$ and $\delta\Omega_{\mathrm{sens}}(t)$ are distinct inputs to that same ensemble. Under a total LO intensity budget, $|\Omega_{\mathrm{com},\mathrm{LO}}|^2=\alpha\bar\Omega_{\mathrm{LO}}^2$ and $|\Omega_{\mathrm{sens},\mathrm{LO}}|^2=(1-\alpha)\bar\Omega_{\mathrm{LO}}^2$, where $\alpha\in[0,1]$ is the communication LO power fraction.

	The above five level MB-RAR model applies when the communication and sensing carriers address two resolvable Rydberg microwave transitions. If the two carriers are closely spaced around the same Rydberg transition, they should be modeled as perturbations of a shared transition RAR. In that case, IF domain separation is possible only when the corresponding intermediate frequencies are distinguishable within the dynamic response bandwidth. If the two carriers also generate the same IF, additional time, code, pilot, or spatial separation is required at the system level.
	
	The interaction between the lasers, the RF fields, and the multi level atomic system is described by the Lindblad master equation \cite{zhu2025general, QS-16-23},
	\begin{equation}
		\frac{\partial\boldsymbol{\rho}}{\partial t}
		=
		\frac{1}{j\hbar}[\mathbf{H}(t),\boldsymbol{\rho}]
		+
		{\mathcal{D}}({\boldsymbol{\rho}}),
		\label{eq:lindblad_master}
	\end{equation}
	where $\boldsymbol{\rho}(t)$ is the density matrix and the Hamiltonian $\mathbf{H}(t)$ is obtained by substituting $\Omega_{\mathrm{com}}(t)$ and $\Omega_{\mathrm{sens}}(t)$ into \eqref{Hamiltonian}. ${\mathcal{D}}({\boldsymbol{\rho}}) = \sum_{\ell=2}^{5}{\mathcal{D}}[L_\ell](\boldsymbol{\rho})$, where for each collapse operator $L_\ell,\ \ell=2,3,4,5$, the corresponding Lindblad dissipator is
	
	\begin{equation}
		{\mathcal{D}}[L_\ell](\boldsymbol{\rho})
		=
		L_\ell\boldsymbol{\rho}L_\ell^\dagger
		-\frac{1}{2}
		\left(
		L_\ell^\dagger L_\ell\boldsymbol{\rho}
		+
		\boldsymbol{\rho}L_\ell^\dagger L_\ell
		\right).
		\label{eq:collapse}
	\end{equation}
	The collapse operators $L_\ell$ are defined as
	\begin{equation}
		L_\ell=
		\begin{cases}
		\sqrt{\Gamma_\ell}\,|\ell-1\rangle\langle\ell|, & \ell=2,3,\\[4pt]
		\sqrt{\Gamma_\ell}\,|3\rangle\langle\ell|, & \ell=4,5,
		\end{cases}
	\label{eq:collapse_operators}
	\end{equation}
	where $\Gamma_\ell(T)=1/\tau_\ell(T)$, $\ell=2,3,4,5$, are the effective population-decay rates of the corresponding energy levels. The state lifetimes are $\tau_2=\tau_{6P_{3/2}}$, $\tau_3=\tau_{38D_{3/2}}$, $\tau_4=\tau_{39P_{3/2}}$, and $\tau_5=\tau_{39P_{1/2}}$, obtained from the Alkali Rydberg Calculator (ARC) atomic model \cite{vsibalic2017arc}. The ground state $|1\rangle$ is treated as stable, and its decay rate is set to $\Gamma_1=0$.
	Substituting these collapse operators \eqref{eq:collapse_operators} into ${\mathcal{D}}({\boldsymbol{\rho}})$ gives the total dissipative contribution \eqref{Dissipative}.
		
	\begin{figure*}
		\centering
		\begin{align}
			\mathbf{H} = \frac{\hbar}{2}
			\begin{bmatrix}
				0 & \Omega_p & 0 & 0 & 0\\
				\Omega_p & -2\Delta_p & \Omega_c & 0 & 0\\
				0 & \Omega_c & -2(\Delta_p + \Delta_c) & \Omega_{\mathrm{com}}^* & \Omega_{\mathrm{sens}}^*\\
				0 & 0 & \Omega_{\mathrm{com}} & -2(\Delta_p + \Delta_c - \Delta_{\mathrm{com}}) & 0\\
				0 & 0 & \Omega_{\mathrm{sens}} & 0 & -2(\Delta_p + \Delta_c - \Delta_{\mathrm{sens}})
			\end{bmatrix}.
			\label{Hamiltonian}
		\end{align}
	\end{figure*}
	
	\begin{figure*}
		\centering
	\begin{equation}
			{\mathcal{D}}({\boldsymbol{\rho}}) = 
			\begin{pmatrix}
				\Gamma_2 \rho_{22} & 
				-\frac{\Gamma_2}{2} \rho_{12} & 
				-\frac{\Gamma_{3}}{2} \rho_{13} & 
				-\frac{\Gamma_{4}}{2} \rho_{14} & 
				-\frac{\Gamma_{5}}{2} \rho_{15} \\
				-\frac{\Gamma_2}{2} \rho_{21} & 
				\Gamma_3 \rho_{33} - \Gamma_2 \rho_{22} & 
				-\frac{\Gamma_2 + \Gamma_{3}}{2} \rho_{23} & 
				-\frac{\Gamma_2 + \Gamma_{4}}{2} \rho_{24} & 
				-\frac{\Gamma_2 + \Gamma_{5}}{2} \rho_{25} \\
				-\frac{\Gamma_{3}}{2} \rho_{31} & 
				-\frac{\Gamma_2 + \Gamma_{3}}{2} \rho_{32} & 
				\Gamma_5 \rho_{55} + \Gamma_4 \rho_{44} - \Gamma_3 \rho_{33} & 
				-\frac{\Gamma_{3} + \Gamma_{4}}{2} \rho_{34} & 
				-\frac{\Gamma_{3} + \Gamma_{5}}{2} \rho_{35} \\
				-\frac{\Gamma_{4}}{2} \rho_{41} & 
				-\frac{\Gamma_2 + \Gamma_{4}}{2} \rho_{42} & 
				-\frac{\Gamma_{3} + \Gamma_{4}}{2} \rho_{43} & 
				-\Gamma_4 \rho_{44} & 
				-\frac{\Gamma_4+\Gamma_5}{2}\rho_{45} \\
				-\frac{\Gamma_{5}}{2} \rho_{51} & 
				-\frac{\Gamma_2 + \Gamma_{5}}{2} \rho_{52} & 
				-\frac{\Gamma_{3} + \Gamma_{5}}{2} \rho_{53} & 
				-\frac{\Gamma_4+\Gamma_5}{2}\rho_{54} & 
				-\Gamma_5 \rho_{55}
			\end{pmatrix}.
			\label{Dissipative}
		\end{equation}
		\hrulefill
	\end{figure*}

For each velocity class, let $\mathbf x(t;v)=\mathrm{vec}(\boldsymbol\rho(t;v))\in\mathbb C^{25}$. With both LO biases retained and both weak inputs zero, $\dot{\mathbf x}=\mathcal L_0(v;\Omega_c,\alpha)\mathbf x$.
The Liouvillian $\mathcal L_0$ includes the common Hamiltonian and dissipator, with the velocity-dependent detunings in \eqref{eq:optical_doppler_detunings}. Its trace-normalized steady state satisfies
\begin{equation}
\mathcal L_0(v)\bar{\mathbf x}(v)=0,\quad \mathrm{Tr}(\bar{\boldsymbol\rho}(v))=\mathbf t^{\mathrm T}\bar{\mathbf x}(v)=1,
\label{eq:liouvillian_working_point}
\end{equation}
where $\bar{\boldsymbol\rho}(v)$ is the steady-state density matrix with $\bar{\mathbf x}(v)=\mathrm{vec}(\bar{\boldsymbol\rho}(v))$, and $\mathbf t\triangleq\mathrm{vec}(\mathbf I_5)$ is the trace-selection vector satisfying $\mathbf t^{\mathrm T}\mathrm{vec}(\mathbf X)=\mathrm{Tr}(\mathbf X)$ for any $5\times5$ matrix $\mathbf X$.
Both RF branches act on this same density matrix. Accordingly, the optical coherence is written as $\rho_{21}(t;v)=\bar\rho_{21}(v)+\delta\rho_{21}(t;v)$. Hence, the shared steady-state coherence is
\begin{equation}
\bar\rho_{21}^{\mathrm D}
=\int_{-\infty}^{\infty}\bar\rho_{21}(v)f_T(v)\,\mathrm dv,
\label{eq:rho21_doppler_average}
\end{equation}
where $\{\cdot\}^{\mathrm D}$ denotes Doppler averaging.
Let $u_b$ denote the Rydberg state, with $u_{\mathrm{com}}=4$ and $u_{\mathrm{sens}}=5$. Then
\begin{equation}
\delta\mathbf H_b(t)
=\frac{\hbar}{2}\left[
\delta\Omega_b(t)|u_b\rangle\langle3|
+\delta\Omega_b^*(t)|3\rangle\langle u_b|
\right].
\label{eq:branch_rf_perturbation_unified}
\end{equation}
The total Hamiltonian perturbation is $\sum_b\delta\mathbf H_b$. Writing $\mathbf x=\bar{\mathbf x}+\delta\mathbf x$ and retaining first-order terms gives
\begin{equation}
\begin{aligned}
\frac{d\delta\mathbf x(t;v)}{dt}
&=\mathcal L_0(v)\delta\mathbf x(t;v)+\sum_b\bigl[\mathcal L_b^{\mathrm u}(v)\bar{\mathbf x}(v)\delta\Omega_b(t)\\
&\hspace{4.5em}+\mathcal L_b^{\mathrm l}(v)\bar{\mathbf x}(v)\delta\Omega_b^*(t)\bigr],
\end{aligned}
\label{eq:liouvillian_small_signal}
\end{equation}
where all branch sums run over $b\in\{\mathrm{com},\mathrm{sens}\}$, and the two source superoperators are
\begin{equation}
\begin{aligned}
\mathcal L_b^{\mathrm u}(v)\mathbf x
&\triangleq\mathrm{vec}\!\left\{-\frac j2[|u_b\rangle\langle3|,\boldsymbol\rho]\right\},\\
\mathcal L_b^{\mathrm l}(v)\mathbf x
&\triangleq\mathrm{vec}\!\left\{-\frac j2[|3\rangle\langle u_b|,\boldsymbol\rho]\right\},
\end{aligned}
\label{eq:branch_perturbation_superoperators}
\end{equation}
where the superscripts $\{\cdot\}^{\mathrm u}$ and $\{\cdot\}^{\mathrm l}$ label the source channels associated with the upper-sideband input and the conjugated lower-sideband input in the positive-frequency response below. 

Let $\mathbf c_{21}$ select the optical coherence, i.e., $\mathbf c_{21}^{\mathrm T}\mathbf x=\rho_{21}$. Fourier transformation of \eqref{eq:liouvillian_small_signal} yields
\begin{equation}
\begin{aligned}
&\delta\mathbf x(\omega;v)
=(j\omega\mathbf I-\mathcal L_0(v))^{-1}\\
&\hspace{1em}\times\sum_b\bigl[\mathcal L_b^{\mathrm u}(v)\bar{\mathbf x}(v)\delta\Omega_b(\omega)+\mathcal L_b^{\mathrm l}(v)\bar{\mathbf x}(v)\delta\Omega_b^*(-\omega)\bigr].
\end{aligned}
\label{eq:rho21_resolvent}
\end{equation}
Here $\mathcal F\{\delta\Omega_b^*(t)\}(\omega)=\delta\Omega_b^*(-\omega)$. For $\omega>0$, the two inputs correspond to frequencies $f_{\mathrm{LO},b}+\omega/(2\pi)$ and $f_{\mathrm{LO},b}-\omega/(2\pi)$, respectively. Define their single-velocity coherence-response kernels as
\begin{equation}
\begin{aligned}
F_b^{\mathrm u}(\omega;v)
&\triangleq\mathbf c_{21}^{\mathrm T}(j\omega\mathbf I-\mathcal L_0(v))^{-1}
\mathcal L_b^{\mathrm u}(v)\bar{\mathbf x}(v),\\
F_b^{\mathrm l}(\omega;v)
&\triangleq\mathbf c_{21}^{\mathrm T}(j\omega\mathbf I-\mathcal L_0(v))^{-1}
\mathcal L_b^{\mathrm l}(v)\bar{\mathbf x}(v).
\end{aligned}
\label{eq:branch_dynamic_response}
\end{equation}
Both kernels are evaluated at the same output frequency $\omega$; the negative-frequency argument belongs to the conjugated input. At other signed frequencies, \eqref{eq:rho21_resolvent} defines the continuation of these two source responses. The resolvent acts on the trace-zero perturbation subspace; at $\omega=0$, the trace-constrained system in Appendix~\ref{app:proof_cross_damping} defines the response. The total optical coherence perturbation is
\begin{equation}
\delta\rho_{21}(\omega;v)\!=\! \sum_b\bigl[F_b^{\mathrm u}(\omega;v)\delta\Omega_b(\omega)\!+\!F_b^{\mathrm l}(\omega;v)\delta\Omega_b^*(-\omega)\bigr].
\label{eq:two_branch_response}
\end{equation}
Their Doppler averages are
\begin{equation}
\begin{aligned}
F_b^{\mathrm{u,D}}(\omega)
&\triangleq\int_{-\infty}^{\infty}F_b^{\mathrm u}(\omega;v)f_T(v)\,\mathrm dv,\\
F_b^{\mathrm{l,D}}(\omega)
&\triangleq\int_{-\infty}^{\infty}F_b^{\mathrm l}(\omega;v)f_T(v)\,\mathrm dv.
\end{aligned}
\label{eq:branch_sideband_response_doppler}
\end{equation}

Write $\delta\Omega_b(t)=\delta\Omega_{b,\mathrm I}(t)+j\delta\Omega_{b,\mathrm Q}(t)$ with real-valued I/Q components. Then $\delta\Omega_b^*(-\omega)=\delta\Omega_{b,\mathrm I}(\omega)-j\delta\Omega_{b,\mathrm Q}(\omega)$. The corresponding coherence responses are
\begin{equation}
\begin{aligned}
G_{b,\mathrm I}^{\mathrm D}(\omega)
&\triangleq F_b^{\mathrm{u,D}}(\omega)+F_b^{\mathrm{l,D}}(\omega),\\
G_{b,\mathrm Q}^{\mathrm D}(\omega)
&\triangleq j\left[F_b^{\mathrm{u,D}}(\omega)-F_b^{\mathrm{l,D}}(\omega)\right].
\end{aligned}
\label{eq:branch_iq_response}
\end{equation}
The functions $G_{b,\mathrm I}^{\mathrm D}$ and $G_{b,\mathrm Q}^{\mathrm D}$ map the two RF coordinates to the same optical coherence. The total thermal response is
\begin{equation}
\delta\rho_{21}^{\mathrm D}(\omega)
=\sum_b\bigl[G_{b,\mathrm I}^{\mathrm D}(\omega)\delta\Omega_{b,\mathrm I}(\omega)+G_{b,\mathrm Q}^{\mathrm D}(\omega)\delta\Omega_{b,\mathrm Q}(\omega)\bigr].
\label{eq:branch_iq_coherence}
\end{equation}

\begin{figure}[!t]
\centering
\subfigure[Resonant.\label{fig:sideband_response_resonant}]{%
\includegraphics[width=0.525\columnwidth]{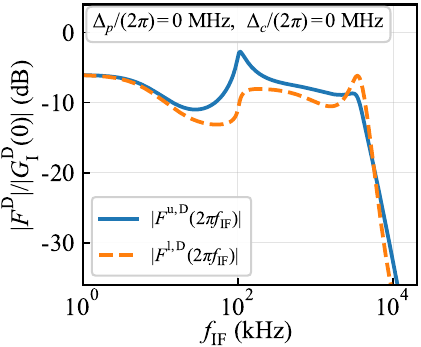}}\hfill
\subfigure[Detuned.\label{fig:sideband_response_detuned}]{%
\includegraphics[width=0.475\columnwidth]{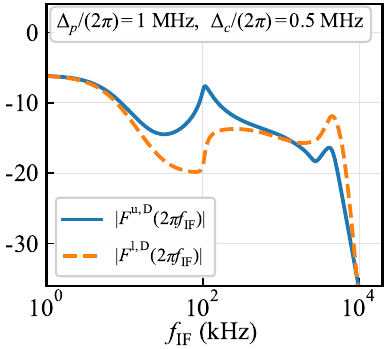}}
\caption{Normalized Doppler-averaged sideband coherence responses: (a) $\Delta_p=\Delta_c=0$; (b) $\Delta_p/(2\pi)=1$~MHz and $\Delta_c/(2\pi)=0.5$~MHz.}
\label{fig:sideband_response}
\end{figure}

The two sideband responses need not be equal, particularly under optical detuning. Figs.~\ref{fig:sideband_response}\subref{fig:sideband_response_resonant} and~\ref{fig:sideband_response}\subref{fig:sideband_response_detuned} compare the communication-sideband kernels under resonant and detuned optical excitation, respectively, using the same $|G_{\mathrm I}^{\mathrm D}(0)|$ normalization within each panel. The two magnitudes nearly coincide at low IF but separate as the IF increases. This separation arises because the direct and conjugate RF perturbations drive different coherence pathways of the LO-biased Liouvillian, whose complex weights depend on frequency. Moreover, optical detuning changes the relative sideband magnitudes and the spectral features by shifting the velocity-dependent resonances. However, resonant excitation alone does not guarantee identical sideband responses over the receive band, and both kernels are retained in \eqref{eq:branch_iq_response}.

Let $\delta\Omega_{b,q}(t)=a_{b,q}\cos(\omega t)$ with $a_{b,q}\in\mathbb R$, $q\in\{\mathrm I,\mathrm Q\}$, and a common phase reference. Each response column is measured by exciting one input at a time, with all other weak inputs zero and both LO biases retained. The shared complex coherence harmonic is defined as
\begin{equation}
	\widehat{\delta\rho}_{21}^{\mathrm D}
	=\lim_{T_{\mathrm{obs}}\to\infty}\frac{2}{T_{\mathrm{obs}}}\int_{t_0}^{t_0+T_{\mathrm{obs}}}
	\delta\rho_{21}^{\mathrm D}(t)e^{-j\omega t}\,\mathrm dt,
\end{equation}
where $T_{\mathrm{obs}}$ is the observation time, taken in simulation as an integer number of modulation periods after transients decay. In the small-signal limit, the real and imaginary parts obey
\begin{equation}
\begin{bmatrix}
\operatorname{Re}\{\widehat{\delta\rho}_{21}^{\mathrm D}\}\\
\operatorname{Im}\{\widehat{\delta\rho}_{21}^{\mathrm D}\}
\end{bmatrix}
=\sum_b
\begin{bmatrix}
\operatorname{Re}\{G_{b,\mathrm I}^{\mathrm D}\}&\operatorname{Re}\{G_{b,\mathrm Q}^{\mathrm D}\}\\
\operatorname{Im}\{G_{b,\mathrm I}^{\mathrm D}\}&\operatorname{Im}\{G_{b,\mathrm Q}^{\mathrm D}\}
\end{bmatrix}
\begin{bmatrix}a_{b,\mathrm I}\\a_{b,\mathrm Q}\end{bmatrix},
\label{eq:branch_real_imag_response}
\end{equation}
where each $2\times2$ block maps one branch's real calibration amplitudes to the common coherence harmonic; two branches give a $2\times4$ calibration map when the blocks are concatenated. The physical absorption response is given as
\begin{equation}
	\mathcal M\{X\}(\omega)
	=\frac{X(\omega)-X^*(-\omega)}{2j}.
\end{equation}
At DC, $\mathcal M\{X\}(0)=\operatorname{Im}X(0)$. Applied to the coherence kernels, it gives
\begin{equation}
		\mathcal F\!\left \{\operatorname{Im}\delta\rho_{21}^{\mathrm D}(t)\right\}(\omega)
		\!=\! \sum_b\!\sum_{q\in\{\mathrm I,\mathrm Q\}}\!
		\mathcal M\!\left\{G_{b,q}^{\mathrm D}\right\}\!(\omega)
		\delta\Omega_{b,q}(\omega).
\end{equation}

\begin{figure}[!t]
\centering
\includegraphics[width=3.3in]{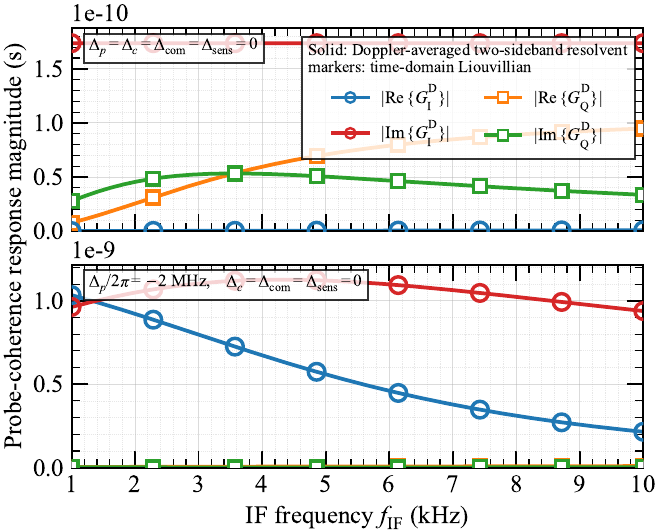}
\caption{Real and imaginary parts of the I/Q coherence-response kernels. Solid: Doppler-averaged double-sideband resolvent; markers: time-domain Liouvillian cosine calibration.}
\label{fig:quadrature_response_validation}
\end{figure}

Fig.~\ref{fig:quadrature_response_validation} compares the real and imaginary parts of the resolvent kernels with coherence harmonics extracted from independent I-only and Q-only time-domain Liouvillian probes. The markers follow the corresponding curves across the tested IF range, verifying that the two-source linearization reproduces the driven coherence response. The component magnitudes are unequal because the I and Q perturbations couple differently to the biased optical coherence. When, in the near-resonance regime, every excited branch has its other three entries small relative to $|\operatorname{Im}G_{b,\mathrm I}^{\mathrm D}|$, the calibration relation reduces to
\begin{equation}
\begin{bmatrix}
\operatorname{Re}\{\widehat{\delta\rho}_{21}^{\mathrm D}\}\\
\operatorname{Im}\{\widehat{\delta\rho}_{21}^{\mathrm D}\}
\end{bmatrix}
\approx\sum_b
\begin{bmatrix}
0&0\\
\operatorname{Im}\{G_{b,\mathrm I}^{\mathrm D}\}&0
\end{bmatrix}
\begin{bmatrix}a_{b,\mathrm I}\\a_{b,\mathrm Q}\end{bmatrix}.
\label{eq:branch_near_resonance_approx}
\end{equation}

The single-velocity I-input kernel is $G_{b,\mathrm I}(\omega;v)=F_b^{\mathrm u}(\omega;v)+F_b^{\mathrm l}(\omega;v)$, whose Doppler average follows from \eqref{eq:branch_sideband_response_doppler} and \eqref{eq:branch_iq_response}.
The scalar performance model uses $\mathcal M\{G_{b,\mathrm I}^{\mathrm D}\}$ in bands where the projected Q-input contribution is negligible. 

\begin{figure}[!t]
	\centering
	\includegraphics[width=3.4in]{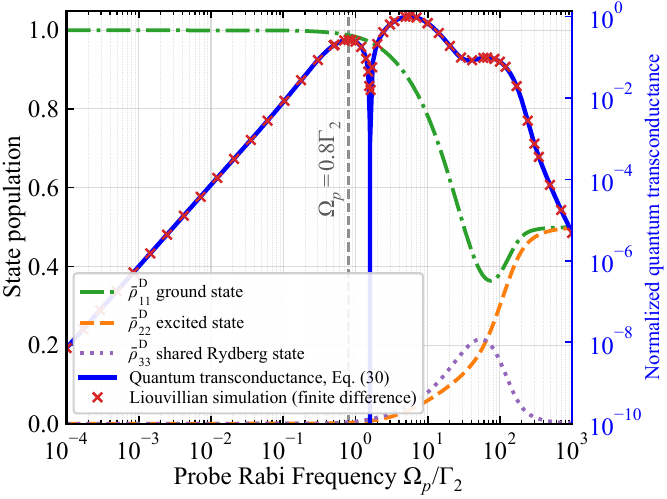}
	\caption{Full five-level Liouvillian response under varying probe Rabi frequency. The left axis shows the state populations. The right logarithmic axis shows the normalized quantum transconductance.}
	\label{fig:probe_response_kernel}
\end{figure}

Fig.~\ref{fig:probe_response_kernel} shows the quantum transconductance and atomic state populations versus probe Rabi frequency. Agreement with central finite differences of the nonlinear steady-state photocurrent validates the Liouvillian RF small-signal response. At low probe strength, most atoms occupy the ground state, and increasing probe power enhances optical readout and transmission. As the probe strengthens, population is transferred to the excited states, modifying the absorption sensitivity of the participating velocity classes. The resulting reduction in net atomic sensitivity eventually offsets the readout enhancement, producing a local transconductance maximum near $\Omega_p=0.8\Gamma_2$. Further excitation weakens the negative response of near-resonant atoms relative to opposite-sign contributions, leading to nearly complete cancellation in the Doppler average at $\Omega_p\approx1.59\Gamma_2$. Beyond this point, the opposite-sign contribution dominates and the conversion recovers, reaching a larger maximum near $\Omega_p=5.57\Gamma_2$ before declining as atomic sensitivity is further suppressed. Meanwhile, the ground-state population decreases, whereas the shared Rydberg-state population initially increases and subsequently falls. In the strong-probe limit, the populations of $|1\rangle$ and $|2\rangle$ approach one half, indicating saturation of the probe transition and diminishing occupation of the Rydberg manifold. Thus, the higher transconductance maximum comes with stronger atomic excitation and greater probe-power requirements. Subsequent simulations select $\Omega_p=0.8\Gamma_2$ to retain appreciable conversion with moderate probe power and approximately $98.8\%$ ensemble-averaged ground-state occupation.

\subsection{Quantum Transconductance}

The shared optical susceptibility follows from the common coherence as
\begin{equation}
\begin{aligned}
\bar\chi^{\mathrm D}
&=-\frac{N|\mu_{21}|^2}{\epsilon_0\hbar}
\frac{\bar\rho_{21}^{\mathrm D}}{\Omega_p/2},\\
\delta\chi^{\mathrm D}(\omega)
&=-\frac{N|\mu_{21}|^2}{\epsilon_0\hbar}
\frac{\delta\rho_{21}^{\mathrm D}(\omega)}{\Omega_p/2},
\end{aligned}
\label{eq:chi_dynamic}
\end{equation}
where $N$ is the atomic number density and $\mu_{21}$ is the optical transition dipole moment. Equation~\eqref{eq:branch_iq_coherence} supplies the contributions of all weak inputs to $\delta\chi^{\mathrm D}$. The minus signs in \eqref{eq:chi_dynamic} follow the positive off-diagonal Rabi convention in \eqref{Hamiltonian}: in the resonant weak-probe two-level limit, $\bar\rho_{21}\simeq-j\Omega_p/\Gamma_2$ for real positive $\Omega_p$, giving $\operatorname{Im}\bar\chi>0$. The steady-state photocurrent is
\begin{equation}
I_{\mathrm{DC}}^{\mathrm D}
=I_{\mathrm{in}}\exp\!\left\{-k_p\ell_{\rm cell}\operatorname{Im}\bar\chi^{\mathrm D}\right\},
\label{eq:IDC_liouvillian}
\end{equation}
where $I_{\mathrm{in}}$ is the input photocurrent reference. This Beer--Lambert readout \cite{jing2020atomic} is evaluated after averaging the coherence over atomic velocities.

At array element $m$, the field-to-Rabi conversion is $\delta\Omega_b^{(m)}(t)=(\mu_b/\hbar)E_{\mathrm{in},b}^{(m)}(t)$, where $\mu_b$ is the microwave transition dipole moment. The branch field envelopes follow \eqref{eq:field_com} and \eqref{eq:field_sens}, expressed relative to their respective LOs.

The two inputs perturb a common photocurrent $i_{\mathrm{out}}$. For the scalar I-input readout used below, let $E_{\mathrm{in},b,\mathrm I}$ be the in-phase branch field envelope. At its operating angular IF $\omega_b=2\pi f_{\mathrm{IF},b}$, define the nonnegative transconductance magnitude
\begin{equation}
\begin{aligned}
g_b(\Omega_c,\alpha;\omega_b)
&\triangleq\left|
\frac{\partial i_{\mathrm{out}}(\omega_b)}
{\ell_{\rm cell}\,\partial E_{\mathrm{in},b,\mathrm I}(\omega_b)}
\right|\\
&=I_{\mathrm{DC}}^{\mathrm D}
\frac{2k_pN|\mu_{21}|^2}{\epsilon_0\hbar|\Omega_p|}
\frac{|\mu_b|}{\hbar}
\left|\mathcal M\!\left\{G_{b,\mathrm I}^{\mathrm D}\right\}(\omega_b)\right|.
\end{aligned}
\label{eq:g_def}
\end{equation}
Thus, $g_b$ identifies the input-to-readout conversion, not a separate density matrix or detector.

\subsection{Noise Model}

		For branch $b\in\{\mathrm{com},\mathrm{sens}\}$, the processed readout is corrupted by three physically distinct noise sources: photon shot noise (PSN) from optical detection, quantum projection noise (QPN) intrinsic to the atomic ensemble, and thermal noise (TN) from the readout electronics. Let $I_{\mathrm{PD}}=I_{\mathrm{DC}}^{\mathrm D}=\mathcal R_{\mathrm{PD}}P_{\mathrm{opt}}$ denote the steady-state photocurrent and $B_{\mathrm{eq},b}$ the equivalent noise bandwidth of branch $b$. Since the three sources arise from independent physical mechanisms, they are modeled as mutually uncorrelated, wide-sense stationary processes. Over $B_{\mathrm{eq},b}$, each is assumed to have a locally flat one-sided power spectral density (PSD) $S$, referred to the same output port across the load $R_{\mathrm{L}}$.
	
	\begin{itemize}
		\item[$\bullet$]\textbf{Photon Shot Noise (PSN):}
		PSN originates from the discrete, Poisson-distributed arrival of photoelectrons at the photodetector (PD), so that the photocurrent fluctuates around its mean $I_{\mathrm{PD}}$ with a white current PSD $2qI_{\mathrm{PD}}$ in $\mathrm{A^2/Hz}$, where $q$ is the elementary charge. Background light and dark current are neglected, since the probe beam dominates the detected optical power. Dissipated in $R_{\mathrm{L}}$, this current fluctuation yields the output PSD $S_{\mathrm{PSN}}=2qI_{\mathrm{PD}}R_{\mathrm{L}}$, and hence
		\begin{equation}
			\sigma_{\mathrm{PSN},b}^2
			=2qI_{\mathrm{PD}}R_{\mathrm{L}}B_{\mathrm{eq},b}
			=2q\mathcal R_{\mathrm{PD}}P_{\mathrm{opt}}R_{\mathrm{L}}B_{\mathrm{eq},b},
			\label{eq:psn_equiv}
		\end{equation}
		which shows that the PSN variance grows linearly with the detected probe power.

		\item[$\bullet$]\textbf{Quantum Projection Noise (QPN):}	
		For an uncorrelated atomic ensemble, we adopt the standard-quantum-limit field-sensitivity model \cite{gong2026rydbergTC}
		
		\begin{equation}
			\mathcal{E}_{\mathrm{QPN},b}
			=\frac{\hbar}{|\mu_b|\sqrt{N_{\mathrm{eff},b}T_{2,b}}},
			\quad
			\bigl[\frac{\mathrm{V/m}}{\sqrt{\mathrm{Hz}}}\bigr],
			\label{eq:qpn_equiv}
		\end{equation}
		where $N_{\mathrm{eff},b}$ is the effective participating atom number, and $T_{2,b}$ is the microwave coherence time. Interpreting $\mathcal E_{\mathrm{QPN},b}^2$ as a one-sided mean-square field-noise PSD, the conversion $|\delta i_{\mathrm{out},b}|=\ell_{\rm cell}g_b|\delta E_b|$ gives
		
		\begin{equation}
		\begin{aligned}
			S_{\mathrm{QPN},b}^{\mathrm{out}}
			&=R_{\mathrm{L}}\ell_{\rm cell}^{2}g_b^{2}\mathcal{E}_{\mathrm{QPN},b}^{2},\\
			\sigma_{\mathrm{QPN},b}^{2}
			&=\frac{R_{\mathrm{L}}\ell_{\rm cell}^{2}g_b^{2}\hbar^{2}}
			{|\mu_b|^{2}N_{\mathrm{eff},b}T_{2,b}}B_{\mathrm{eq},b}.
		\end{aligned}
			\label{eq:qpn_psd}
		\end{equation}

		\item[$\bullet$]\textbf{Thermal Noise (TN):}
		TN comprises the Johnson--Nyquist noise of the load resistor and the excess noise added by the post-detection amplifier chain. Assuming that the readout electronics operate at the same temperature $T$ as the vapor cell and referring both contributions to the output port through the amplifier noise factor $F_{\mathrm{amp}}$, which absorbs the chain gain, gives the available noise PSD $S_{\mathrm{TN}}=k_BTF_{\mathrm{amp}}$. Thus,
		\begin{equation}
			\sigma_{\mathrm{TN},b}^2
			=k_BTF_{\mathrm{amp}}B_{\mathrm{eq},b}.
			\label{eq:tn_equiv}
		\end{equation}
		Unlike PSN and QPN, $\sigma_{\mathrm{TN},b}^2$ is independent of the optical power and of the LO allocation. It therefore sets a fixed electronic noise floor that dominates at low probe power.
	\end{itemize}
	
	Because the three contributions in \eqref{eq:psn_equiv}, \eqref{eq:qpn_psd}, and \eqref{eq:tn_equiv} are mutually uncorrelated and referred to the same output port, their variances add. After photodetection and band-limited baseband processing, the residual noise aggregates many independent microscopic fluctuations. By the central limit theorem, it is well approximated as circularly symmetric complex Gaussian
	\begin{equation}
		\tilde n_b\sim\mathcal{CN}\bigl(0,\sigma_{n,b}^2\bigr),\quad
		\sigma_{n,b}^2
		=\sigma_{\mathrm{PSN},b}^2+\sigma_{\mathrm{QPN},b}^2+\sigma_{\mathrm{TN},b}^2,
		\label{eq:total_noise_variance}
	\end{equation}
	where this Gaussian model applies to the processed baseband residual.

	\section{Quantum Tradeoff Analysis}
	\label{sec:tradeoffs}
	This section reveals two intrinsic tradeoffs in RAR-ISAC: the coupling-laser tradeoff and IBGI.
	\begin{figure}[!t]
		\centering
		\includegraphics[width=\columnwidth]{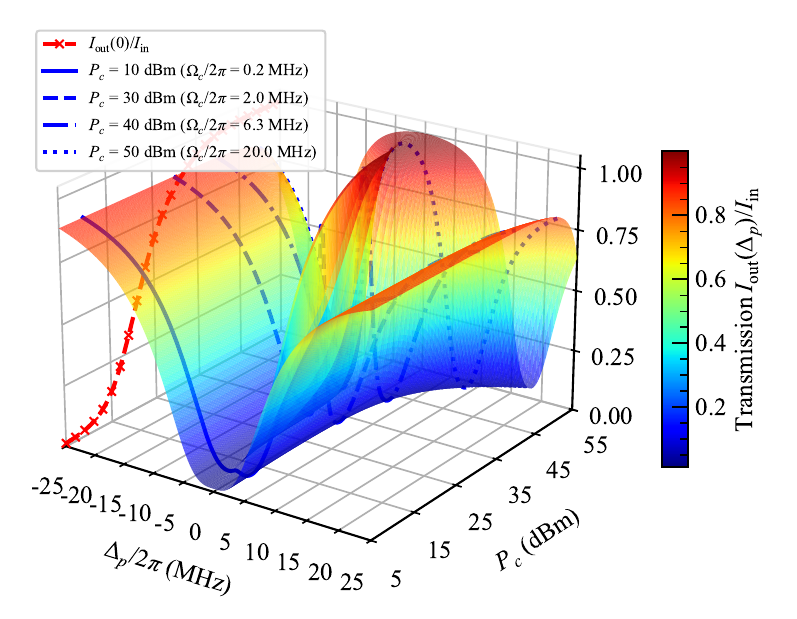}
		\caption{Steady-state probe-transmission spectra versus detuning for different coupling powers $P_c$.}
		\label{fig:EIT}
	\end{figure}
	\subsection{Coupling-Laser Tradeoff and Effective Response Bandwidth}
	\label{sec:dynamic_bandwidth}
	The recoverable communication bandwidth is determined by the frequency response from a time-varying RF perturbation to the post-detection optical readout. The incident communication signal induces a weak time-varying perturbation of the microwave Rabi frequency, which is converted into a baseband photocurrent through the Doppler-averaged small-signal atomic response $G_{\mathrm{com},\mathrm I}^{\mathrm D}(\omega)$.

	Fig.~\ref{fig:EIT} compares the steady-state probe-transmission spectra for different coupling powers. The transmission features widen and their separation grows as the coupling power increases. The coupling field first changes the interference between optical excitation pathways that produces EIT; at stronger coupling, it dresses the coupled levels and leads to resolved AT structure. The resulting peak separation follows the coupling strength. These features describe a static detuning scan; the communication bandwidth additionally depends on how rapidly a weak RF perturbation is converted into the detected optical signal.

For the near-resonant in-phase absorption readout defined by $\mathcal M$ in Section~II-B, the normalized dynamic RF-to-optical transfer function, with $f$ in Hz, is
	\begin{equation}
		\mathcal{H}_{\mathrm{com}}(f;\Omega_c,\alpha)
		\triangleq
		\frac{
		\mathcal M\!\left\{G_{\mathrm{com},\mathrm I}^{\mathrm D}\right\}(2\pi f;\Omega_c,\alpha)
		}{
		\left|\mathcal M\!\left\{G_{\mathrm{com},\mathrm I}^{\mathrm D}\right\}(0;\Omega_c,\alpha)\right|}.
		\label{eq:normalized_dynamic_response}
	\end{equation}
	The effective bandwidth $B_{\mathrm{com}}(\Omega_c,\alpha)$ is defined as the $-3$~dB bandwidth of the dynamic response in \eqref{eq:normalized_dynamic_response}, i.e., the smallest frequency $f$ at which $|\mathcal{H}_{\mathrm{com}}(f;\Omega_c,\alpha)|=1/\sqrt{2}$. It is evaluated numerically from the Doppler-averaged response over $0\le f\le100$~MHz at each operating point.

	\begin{figure}[!t]
		\centering
		\includegraphics[width=2.8in]{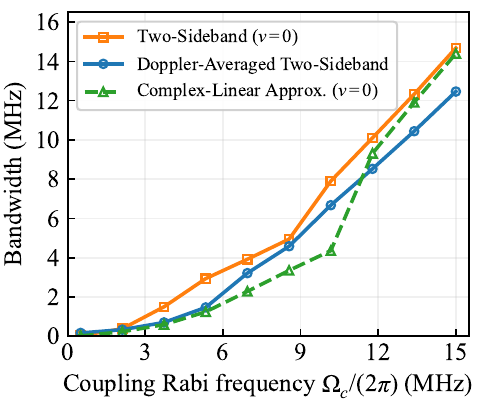}
		\caption{Effective bandwidth of the dynamic response versus coupling Rabi frequency under resonant excitation. The Doppler-averaged double-sideband curve gives $B_{\mathrm{com}}$ at $T=300$~K.}
		\label{fig:bandwidth_vs_coupling}
	\end{figure}

	Fig.~\ref{fig:bandwidth_vs_coupling} compares the effective dynamic bandwidths of the Doppler-averaged and zero-velocity double-sideband responses with the zero-velocity complex-linear reference. The two zero-velocity curves also differ: the complex-linear reference uses $|\operatorname{Im}G_{\mathrm{com},\mathrm I}(2\pi f;0)|$, whereas the physical absorption projection combines the positive and negative frequency responses before taking the magnitude. The bandwidth increases over the displayed coupling sweep because the stronger coupling field changes the coherent evolution rates of the common EIT channel, allowing the readout to follow faster RF perturbations. Doppler averaging combines complex responses with velocity-dependent detunings and reshapes the normalized spectrum.
	
\subsection{Inter Band Gain Inhibition}
	\label{sec:ibgi}
	\label{sec:ibgi_numerical}
	
	Both RF branches couple to the shared Rydberg state $|3\rangle$. Their LO biases therefore jointly determine the atomic steady state and the RF-to-optical response. We refer to the reduction in one branch's transconductance caused by strengthening the other LO bias while holding the own-branch bias fixed as IBGI.

	Let $U_b\triangleq|\Omega_{b,\mathrm{LO}}|$ denote the LO Rabi amplitude of branch $b\in\{\mathrm{com},\mathrm{sens}\}$. Under the total LO budget, $U_{\mathrm{com}}=\bar\Omega_{\mathrm{LO}}\sqrt{\alpha}$ and $U_{\mathrm{sens}}=\bar\Omega_{\mathrm{LO}}\sqrt{1-\alpha}$. The LO contributions to the shared-state coherence denominator are written as
	\begin{equation}
		\Gamma_{3,\mathrm{eff}}(\alpha)
		=\frac{\Gamma_3}{2}
		+\Gamma_{\mathrm{inj}}^{\mathrm{com}}(\alpha)
		+\Gamma_{\mathrm{inj}}^{\mathrm{sens}}(1-\alpha),
		\label{eq:gamma_total_inj}
	\end{equation}
	where
	\begin{equation}
		\begin{aligned}
		\Gamma_{\mathrm{inj}}^{\mathrm{com}}(\alpha)
		&=\frac{\alpha\bar\Omega_{\mathrm{LO}}^2\Gamma_4}
		{8(\Gamma_4^2/4+\Delta_{\mathrm{com}}^2)},\\
		\Gamma_{\mathrm{inj}}^{\mathrm{sens}}(1-\alpha)
		&=\frac{(1-\alpha)\bar\Omega_{\mathrm{LO}}^2\Gamma_5}
		{8(\Gamma_5^2/4+\Delta_{\mathrm{sens}}^2)}.
		\end{aligned}
		\label{eq:gamma_inj_terms}
	\end{equation}
	Increasing $\alpha$ strengthens the communication contribution and reduces the sensing contribution. To separate the explicit LO amplitude from this response dependence, define $\mathcal A_b\triangleq g_b/U_b$, the transconductance per unit LO Rabi amplitude. For fixed $\Omega_c$, detunings, LO phases, and branch IFs, let
	\begin{equation}
		\eta_{bb'}\triangleq
		\frac{U_{b'}}{\mathcal A_b}
		\frac{\partial\mathcal A_b}{\partial U_{b'}},
		\qquad b,b'\in\{\mathrm{com},\mathrm{sens}\},
		\label{eq:branch_gain_elasticities}
	\end{equation}
	where each partial derivative holds the other LO amplitude fixed. The dimensionless coefficients $\eta_{bb}$ and $\eta_{bb'}$, $b'\ne b$, quantify the own-branch and cross-branch changes in the response per unit LO amplitude, respectively.

	\begin{lemma}[LO Allocation and Branch Gain]
		\label{lem:cross_damping}
		At a differentiable operating point with $0<\alpha<1$ and positive branch transconductances, the communication and sensing gains satisfy
		\begin{equation}
			\frac{\partial_\alpha g_{\mathrm{com}}}{g_{\mathrm{com}}}
			=\frac{1}{2\alpha}
			+\frac{\eta_{\mathrm{com},\mathrm{com}}}{2\alpha}
			-\frac{\eta_{\mathrm{com},\mathrm{sens}}}{2(1-\alpha)},
			\label{eq:communication_gain_competition}
		\end{equation}
		\begin{equation}
			\frac{\partial_\alpha g_{\mathrm{sens}}}{g_{\mathrm{sens}}}
			=-\frac{1}{2(1-\alpha)}
			-\frac{\eta_{\mathrm{sens},\mathrm{sens}}}{2(1-\alpha)}
			+\frac{\eta_{\mathrm{sens},\mathrm{com}}}{2\alpha}.
			\label{eq:cross_damping}
		\end{equation}
	\end{lemma}
	\begin{IEEEproof}
		The derivation and the calculation of $\eta_{bb'}$ from the LO-biased Liouvillian are provided in Appendix~\ref{app:proof_cross_damping}.
	\end{IEEEproof}

	The first term in \eqref{eq:communication_gain_competition} is the direct mixing enhancement associated with the communication LO amplitude. The second term accounts for the accompanying change in the response per unit LO, including self broadening and population redistribution. A negative $\eta_{\mathrm{com},\mathrm{com}}$ reduces the direct enhancement, with their combined contribution changing sign at $\eta_{\mathrm{com},\mathrm{com}}=-1$. The third term captures the effect of reducing the sensing LO; it increases the communication gain when $\eta_{\mathrm{com},\mathrm{sens}}<0$.

	For sensing, \eqref{eq:cross_damping} separates direct LO loss, the own-branch response change, and the influence of the communication LO. A negative $\eta_{\mathrm{sens},\mathrm{sens}}$ gives gain recovery as the sensing LO is reduced, whereas $\eta_{\mathrm{sens},\mathrm{com}}<0$ yields the IBGI penalty. These coefficients include changes in the steady-state photocurrent and the RF-to-optical conversion slope. The own-branch term therefore includes self broadening and population redistribution, while a negative cross-branch coefficient identifies local IBGI; their signs are determined by the operating point.

	IBGI acts through the atomic operating point and branch transconductances. It can therefore couple spectrally separated RF branches that share $|3\rangle$, even when their waveforms are separated at the readout. Figs.~\ref{fig:ibgi_allocation_mechanism} and~\ref{fig:ibgi_coupling_control} illustrate the gain competition in Lemma~\ref{lem:cross_damping}.

\begin{figure}[!t]
\centering
\includegraphics[width=\columnwidth]{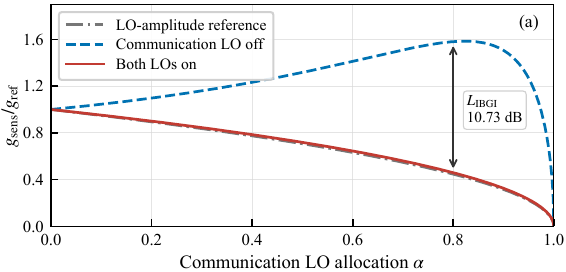}
\par\vspace{1mm}
\includegraphics[width=\columnwidth]{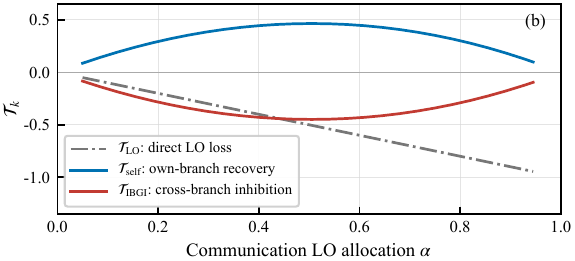}
\caption{IBGI under LO allocation at $\Omega_c/(2\pi)=6$~MHz. (a) Normalized sensing gains with both LOs on and with the communication LO off, compared with direct LO-amplitude scaling. (b) Direct LO loss, own-branch recovery, and cross-branch inhibition in \eqref{eq:cross_damping}.}
\label{fig:ibgi_allocation_mechanism}
\end{figure}

Fig.~\ref{fig:ibgi_allocation_mechanism}(a) illustrates IBGI during LO allocation. At the same sensing bias, the simultaneous gain remains below the communication-LO-off reference, with a $10.73$~dB penalty at $\alpha=0.8$. The communication LO couples the shared state $|3\rangle$ to the damped state $|4\rangle$, adding a contribution to the common coherence denominator and suppressing the optical response to the sensing perturbation. This means that the sensing conversion efficiency is also affected by the communication bias, even when its own LO remains fixed.

Fig.~\ref{fig:ibgi_allocation_mechanism}(b) separates the allocation sensitivity in \eqref{eq:cross_damping} into direct LO loss, own-branch recovery, and cross-branch inhibition. Multiplying each term by the common positive factor $2\alpha(1-\alpha)$ for $0<\alpha<1$ gives the dimensionless contributions $\mathcal T_{\mathrm{LO}}=-\alpha,\ \mathcal T_{\mathrm{self}}=-\alpha\eta_{\mathrm{sens},\mathrm{sens}},\ \mathcal T_{\mathrm{IBGI}}=(1-\alpha)\eta_{\mathrm{sens},\mathrm{com}}.$
Their sum equals $2\alpha(1-\alpha)\partial_\alpha\ln g_{\mathrm{sens}}^{\mathrm{joint}}$. Reducing the sensing LO alleviates its own broadening, while strengthening the communication LO introduces a comparable inhibition through the shared coherence. Consequently, IBGI largely offsets the response recovery available from reducing the sensing bias, leaving the simultaneous gain approximately proportional to $\sqrt{1-\alpha}$. This near-amplitude scaling therefore reflects compensation between two substantial atomic-response contributions.

\begin{figure}[!t]
\centering
\includegraphics[width=3in]{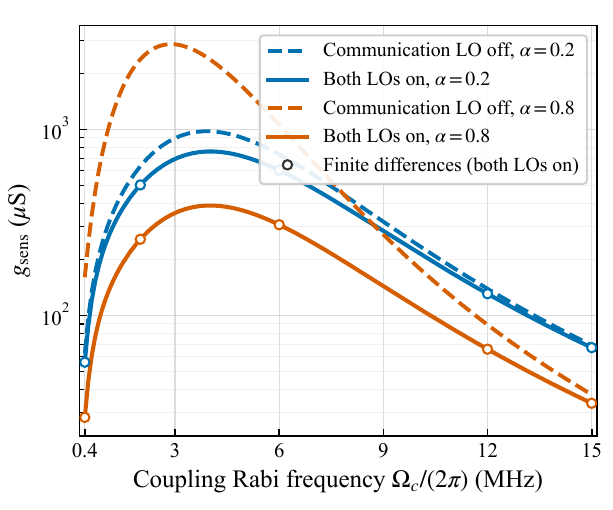}
\caption{Sensing transconductance versus coupling Rabi frequency for $\alpha=0.2$ and $0.8$. Solid curves show both LOs on; dashed curves show the communication-LO-off reference at the same sensing bias. Open circles denote independent finite-difference validation of the simultaneous-operation gains.}
\label{fig:ibgi_coupling_control}
\end{figure}

Fig.~\ref{fig:ibgi_coupling_control} illustrates the dependence of IBGI on optical coupling strength for $\alpha=0.2$ and $0.8$. Switching on the communication LO reverses the allocation ordering over the lower-to-intermediate coupling range and shifts the sensing-gain maximum toward stronger coupling. The added microwave coupling changes the shared coherence and its balance with the optical drive. Thus, an optical setting optimized for isolated sensing need not maximize the simultaneous sensing gain.
At stronger coupling, optical dressing reduces the probe response to both microwave branches. The weaker response to the communication bias reduces relative IBGI, while the weaker response to the sensing perturbation lowers the useful conversion gain. Therefore, minimizing the relative IBGI penalty does not maximize the sensing transconductance, and optical coupling and LO allocation must be selected according to their combined effect on the atomic response.

\FloatBarrier

	\section{C-R Region Characterization}    
	\label{sec:C-R}
	This section formulates the communication rate and AoA CRB, defines the achievable C-R region, and presents its joint optimization and endpoint characterization.
	
	\subsection{Performance Metrics}
	\label{sec:performance}

	For an $M$-element RAR array, the post detection complex baseband communication model is given by
	\begin{equation}
		\mathbf{y}_{\mathrm{com}}
		=
		\sqrt{\frac{\eta_0P_{\mathrm{tx}}R_{\mathrm{L}}}{A_{\mathrm{e}}}}
		\ell_{\rm cell}\mathbf{G}_{\mathrm{com}}\mathbf{h}_{\mathrm{com}}x_{\mathrm{com}}
		+
		\mathbf{w}_{\mathrm{com}},
		\label{eq:com_signal_model}
	\end{equation}
where $\mathbf{y}_{\mathrm{com}}\in\mathbb{C}^{M\times1}$, $x_{\mathrm{com}}$ is the normalized communication symbol with $\mathbb{E}[|x_{\mathrm{com}}|^2]=1$, $\mathbf{h}_{\mathrm{com}}$ is defined in the RF signal model, and $R_{\mathrm{L}}$ is the readout load resistance. The transconductance matrix is $\mathbf{G}_{\mathrm{com}}=g_{\mathrm{com}}(\Omega_c,\alpha;\omega_{\mathrm{com}})\mathbf{I}_M$, where $g_{\mathrm{com}}$ is obtained from \eqref{eq:g_def}. The noise vector is modeled as $\mathbf{w}_{\mathrm{com}}\sim\mathcal{CN}(\mathbf{0},\sigma_{n,\mathrm{com}}^2\mathbf{I}_M)$, with $\sigma_{n,\mathrm{com}}^2$ specified in \eqref{eq:total_noise_variance}.
	
	With maximal ratio combining at the BS, the effective communication SNR is
	\begin{equation}
		\gamma_{\mathrm{com}}(\Omega_c,\alpha)
		=
		\frac{
			\eta_0P_{\mathrm{tx}}R_{\mathrm{L}}\ell_{\rm cell}^2
			g_{\mathrm{com}}^2(\Omega_c,\alpha;\omega_{\mathrm{com}})
			\|\mathbf{h}_{\mathrm{com}}\|^2
		}{
			A_{\mathrm{e}}\sigma_{n,\mathrm{com}}^2
		}.
		\label{eq:snr_com}
	\end{equation}
	The achievable communication rate is then given by
	\begin{equation}
		R(\Omega_c,\alpha)
		=
		B_{\mathrm{com}}(\Omega_c,\alpha)
		\log_2
		\left(
		1+\gamma_{\mathrm{com}}(\Omega_c,\alpha)
		\right),
		\label{eq:rate_with_sigma_total}
	\end{equation}
	
	where $B_{\mathrm{com}}(\Omega_c,\alpha)$ is the effective bandwidth in Section~\ref{sec:dynamic_bandwidth}.
	The LO allocation affects the rate through the branch gain, dynamic bandwidth, and readout noise. At a differentiable operating point with $g_{\mathrm{com}}>0$, differentiation of \eqref{eq:snr_com} and \eqref{eq:rate_with_sigma_total} gives
	\begin{equation}
		\begin{aligned}
		\partial_\alpha R
		&=\frac{\partial_\alpha B_{\mathrm{com}}}{\ln 2}
		\ln(1+\gamma_{\mathrm{com}})\\
		&\quad+\frac{B_{\mathrm{com}}\gamma_{\mathrm{com}}}
		{(1+\gamma_{\mathrm{com}})\ln 2}
		\left(\frac{2\partial_\alpha g_{\mathrm{com}}}{g_{\mathrm{com}}}
		-\partial_\alpha\ln\sigma_{n,\mathrm{com}}^2\right).
		\end{aligned}
		\label{eq:rate_allocation_derivative}
	\end{equation}

	For sensing, the receiver estimates the target AoA $\theta$ from the radar echoes. For $K$ snapshots, the post detection sensing observation at the $k$-th snapshot is modeled as
	\begin{equation}
		\begin{aligned}
		\mathbf{y}_{\mathrm{sens},k}
		&=
		\sqrt{\frac{\eta_0P_{\mathrm{echo}}R_{\mathrm{L}}}{A_{\mathrm{e}}}}
		\ell_{\rm cell}g_{\mathrm{sens}}(\Omega_c,\alpha;\omega_{\mathrm{sens}})\\
		&\quad\times\mathbf{a}(\theta)s_k+\mathbf{w}_{\mathrm{sens},k},
		\end{aligned}
		\label{eq:sens_signal_model}
	\end{equation}
	where $s_k$ is the unknown deterministic complex echo coefficient at the $k$-th snapshot with normalized power $|s_k|=1$, $k=1,\ldots,K$, the noise vectors are independent across snapshots, and $\mathbf{a}(\theta)$ is the array spatial vector defined in \eqref{eq:field_sens}. The sensing branch transconductance $g_{\mathrm{sens}}$ is evaluated from \eqref{eq:g_def}, and the noise vector satisfies $\mathbf{w}_{\mathrm{sens},k}\sim\mathcal{CN}(\mathbf{0},\sigma_{n,\mathrm{sens}}^2\mathbf{I}_M)$, where $\sigma_{n,\mathrm{sens}}^2$ is given in \eqref{eq:total_noise_variance}.
			
	\begin{lemma}
		\label{lem:CRB_AoA}
		For RAR-based ISAC systems, the AoA CRB is given by
		\begin{equation}
			\mathrm{CRB}(\theta,\Omega_c,\alpha)
			=
			\frac{1}{
				2K
				\gamma_{\mathrm{sens}}(\Omega_c,\alpha)
				\|\mathbf{P}_{\perp}(\theta)\dot{\mathbf{a}}(\theta)\|^2
			},
			\label{eq:CRB_sensing}
		\end{equation}
		where $\gamma_{\mathrm{sens}} = {\eta_0P_{\mathrm{echo}}R_{\mathrm{L}}\ell_{\rm cell}^2g_{\mathrm{sens}}^2}/{A_{\mathrm{e}}\sigma_{n,\mathrm{sens}}^2}$ is the dimensionless sensing SNR, and $\mathbf P_{\perp}=\mathbf I_M-\mathbf a\mathbf a^{\mathrm H}/M$ projects orthogonally to the steering vector $\mathbf a(\theta)$.	
		For a ULA with spacing $d$, \eqref{eq:CRB_sensing} reduces to
		\begin{equation}
			\mathrm{CRB}(\theta)
			\!=\!
			\frac{
				6A_{\mathrm{e}}\sigma_{n,\mathrm{sens}}^2
			}{
				K\eta_0P_{\mathrm{echo}}R_{\mathrm{L}}\ell_{\rm cell}^2
				g_{\mathrm{sens}}^2 M(M^2-1)(kd\cos\theta)^2
			}.
		\label{eq:CRB_sensing_final}
		\end{equation}
	\end{lemma}
	\begin{IEEEproof}
		Its derivation is provided in Appendix~\ref{app:proof_CRB_AoA}.
	\end{IEEEproof}

For the multi-target load evaluation, consider $L$ targets at angles $\theta_l$, $l=1,\ldots,L$, whose returns can be separated without mutual interference. Each separated observation follows \eqref{eq:sens_signal_model} with the same $K$ snapshots and branch noise model. Keeping the total echo power $P_{\mathrm{echo}}$ fixed and assigning equal powers $P_{\mathrm{echo}}/L$, \eqref{eq:CRB_sensing_final} gives the average marginal AoA bound

\begin{equation}
	\overline C_L(\Omega_c,\alpha)
	=\frac{1}{L}\sum_{l=1}^{L}
	\frac{L C_0(\Omega_c,\alpha)}{\cos^2\theta_l},
\end{equation}

where $C_0(\Omega_c,\alpha)$ is the single-target broadside CRB evaluated with the full power $P_{\mathrm{echo}}$. All targets share the existing sensing branch and its atomic operating point. 
	
\subsection{Physically Feasible Pareto Front}
	For each candidate $\boldsymbol{\psi}$, the physically feasible space $\mathcal {F}$ is defined as
	\begin{equation}
		\mathcal{F} = \left\{ (\Omega_c, \alpha) \left|
		\begin{array}{l}
			\Omega_c^{\min} \le \Omega_c \le \Omega_c^{\max}, \\
			0 \le \alpha \le 1
		\end{array}
		\right. \right\},
	\end{equation}
	where $\Omega_c^{\min}/2\pi=0.4~\mathrm{MHz}$ and $\Omega_c^{\max}/2\pi=15~\mathrm{MHz}$. This range spans the weak-EIT and AT spectral regimes.
	
	The ISAC design aims to maximize the communication rate and minimize the sensing estimation error. This yields the multi objective optimization problem
	\begin{equation}
		\min_{\boldsymbol{\psi} \in \mathcal{F}} \quad \mathbf{F}(\boldsymbol{\psi}) = \left[ -R(\boldsymbol{\psi}), \, \mathrm{CRB}(\boldsymbol{\psi}) \right]^T.
		\label{eq:MOO}
	\end{equation}
	The solution of \eqref{eq:MOO} is the Pareto boundary of the achievable C-R region, denoted by $\mathcal{C}_{C\text{-}R}$. A point $\boldsymbol{\psi}^{\star}$ is Pareto optimal if there is no other feasible $\boldsymbol{\psi}\in\mathcal{F}$ satisfying $R(\boldsymbol{\psi})\ge R(\boldsymbol{\psi}^{\star})$ and $\mathrm{CRB}(\boldsymbol{\psi})\le\mathrm{CRB}(\boldsymbol{\psi}^{\star})$, with at least one strict inequality.

	To characterize the Pareto boundary, we adopt the $\epsilon$-constraint formulation \cite{mavrotas2009epsilon}. For the fixed target angle $\theta$, write $C(\boldsymbol\psi)\triangleq\mathrm{CRB}(\theta,\Omega_c,\alpha)$. Given a required rate $R_{\mathrm{req}}$, we first minimize the sensing CRB:
	\begin{subequations}\label{prob:P_constrained}
		\begin{align}
		\min_{\boldsymbol\psi}\quad &C(\boldsymbol\psi)\\
		\mathrm{s.t.}\quad &R(\boldsymbol\psi)\ge R_{\mathrm{req}},\label{eq:rate_const_phys}\\
		&\boldsymbol\psi\in\mathcal F.
		\end{align}
	\end{subequations}
	
	Let $C^\star(R_{\mathrm{req}})$ be its optimal value and define the minimizer set of \eqref{prob:P_constrained} as
	\begin{equation}
		\begin{aligned}
		\mathcal S(R_{\mathrm{req}})\triangleq
		\{\boldsymbol\psi\in\mathcal F:\;&R(\boldsymbol\psi)\ge R_{\mathrm{req}},\\
		&C(\boldsymbol\psi)=C^\star(R_{\mathrm{req}})\}.
		\end{aligned}
		\label{eq:crb_minimizer_set}
	\end{equation}
	Then, among these minimizers, we select the one with the largest communication rate:
	\begin{equation}
		\boldsymbol\psi^\star(R_{\mathrm{req}})\in
		\underset{\boldsymbol\psi\in\mathcal S(R_{\mathrm{req}})}{\arg\max}
		R(\boldsymbol\psi).
		\label{eq:pareto_rate_tiebreak}
	\end{equation}

	\begin{proposition}[Pareto Optimality of the $\epsilon$-Constraint Solution]
		\label{prop:tightness}
		For a feasible rate requirement, suppose that the optima of \eqref{prob:P_constrained} and \eqref{eq:pareto_rate_tiebreak} are attained. Every resulting design $\boldsymbol\psi^\star(R_{\mathrm{req}})$ is Pareto optimal. Conversely, every Pareto-optimal design $\boldsymbol\psi_{\mathrm P}$ solves both \eqref{prob:P_constrained} and \eqref{eq:pareto_rate_tiebreak} for $R_{\mathrm{req}}=R(\boldsymbol\psi_{\mathrm P})$.
	\end{proposition}
	\begin{IEEEproof}
		Suppose a feasible design $\widetilde{\boldsymbol\psi}$ dominates $\boldsymbol\psi^\star$. Then $R(\widetilde{\boldsymbol\psi})\ge R(\boldsymbol\psi^\star)\ge R_{\mathrm{req}}$. If $C(\widetilde{\boldsymbol\psi})<C(\boldsymbol\psi^\star)$, the optimality of \eqref{prob:P_constrained} is contradicted. Otherwise, dominance requires equal CRBs and $R(\widetilde{\boldsymbol\psi})>R(\boldsymbol\psi^\star)$, which contradicts the maximality in \eqref{eq:pareto_rate_tiebreak}. Hence, $\boldsymbol\psi^\star$ is Pareto optimal.

		Conversely, set $R_{\mathrm{req}}=R(\boldsymbol\psi_{\mathrm P})$ for a Pareto-optimal design $\boldsymbol\psi_{\mathrm P}$. A smaller CRB in \eqref{prob:P_constrained} would produce a dominating design, so $C^\star(R_{\mathrm{req}})=C(\boldsymbol\psi_{\mathrm P})$. A larger rate in \eqref{eq:pareto_rate_tiebreak} would also dominate $\boldsymbol\psi_{\mathrm P}$. Thus, $\boldsymbol\psi_{\mathrm P}$ solves both \eqref{prob:P_constrained} and \eqref{eq:pareto_rate_tiebreak}.
	\end{IEEEproof}

\begin{algorithm}[t]
\caption{Rate constrained joint optimization (RCJO)}
\label{alg:RCJO}
\small
\begin{algorithmic}[1]
\REQUIRE Descending $\{R_{\rm req}^{(k)}\}_{k=1}^{N_R}$, $\mathcal Z_0$, $\mathcal P_0$, $N_{\rm start}$, $\Delta_0$, $\Delta_{\min}$, $I_{\max}$, $\kappa$, $c_{\min}$.
\STATE $\mathcal K\gets\varnothing$, $\mathcal P\gets\mathcal P_0$, $\boldsymbol z^{\rm prev}\gets\boldsymbol z^{\rm F}$.
\FOR{$k=1,\ldots,N_R$}
\STATE Set $r=R_{\rm req}^{(k)}$; select starts $\mathcal Z_k$; skip if empty.
\FOR{each start; reset $(\boldsymbol z,\Delta,i)\gets(\boldsymbol z^{(0)},\Delta_0,0)$}
\WHILE{$\Delta>\Delta_{\min}$ and $i<I_{\max}$}
\STATE $\widehat{\boldsymbol z}\gets\arg\min_{\boldsymbol v\in\mathcal V_r(\boldsymbol z,\Delta)}\widetilde C(\boldsymbol v)$; $i\gets i+1$.
\STATE $\boldsymbol z\gets\widehat{\boldsymbol z}$ if sufficient decrease; else $\Delta\gets\Delta/2$.
\ENDWHILE
\STATE Store $\boldsymbol z$ and its stopping status in $\mathcal Q_k$.
\ENDFOR
\STATE Set $\boldsymbol z^{\rm prev}$ to the best result in $\mathcal Q_k$; append its pair to $\mathcal P$.
\ENDFOR
\RETURN $\operatorname{ND}(\mathcal P)$ and recorded statuses.
\end{algorithmic}
\end{algorithm}

	The difficulty in solving~\eqref{prob:P_constrained} lies in the fact that $R$ and $C$ are available only through numerical evaluation of the Liouvillian steady state and its RF small-signal linearization. Moreover, $R$ is not necessarily monotonic in $\alpha$. However, since the control vector is only two-dimensional and the rate constraint can be verified pointwise, \eqref{prob:P_constrained} can be solved by a derivative-free search confined to the feasible region. To balance the two controls, we normalize them as
	\begin{equation}
		\boldsymbol\psi(\boldsymbol z)
		=\bigl(\Omega_c^{\min}+(\Omega_c^{\max}-\Omega_c^{\min})z_1,\ z_2\bigr),
		\quad \boldsymbol z\in[0,1]^2,
	\end{equation}
	and write $\widetilde R\triangleq R\circ\boldsymbol\psi$ and $\widetilde C\triangleq C\circ\boldsymbol\psi$. For a rate requirement $r$, the feasible neighborhood of $\boldsymbol z$ at scale $\Delta$ is defined along the coordinate and diagonal directions $\mathcal D=\{(\pm1,0),(0,\pm1),(\pm1,\pm1)\}$ as
	\begin{equation}
		\begin{aligned}
			\mathcal V_r(\boldsymbol z,\Delta)=\{\boldsymbol z\}\cup\{&\boldsymbol v=\Pi_{[0,1]^2}(\boldsymbol z+\Delta\boldsymbol d):
			\boldsymbol d\in\mathcal D,\\
			&\widetilde R(\boldsymbol v)\ge r,\quad
			0<\widetilde C(\boldsymbol v)<\infty\},
		\end{aligned}
	\end{equation}
	where $\Pi_{[0,1]^2}(\cdot)$ denotes the projection onto the box. Minimizing $\widetilde C$ over $\mathcal V_r(\boldsymbol z,\Delta)$ yields a candidate $\widehat{\boldsymbol z}$. The candidate replaces $\boldsymbol z$ if
	\begin{equation}
		\widetilde C(\widehat{\boldsymbol z})<\widetilde C(\boldsymbol z)-\kappa\max\bigl\{\widetilde C(\boldsymbol z),c_{\min}\bigr\},
	\end{equation}
	where $\kappa$ is a relative tolerance and $c_{\min}$ an absolute floor. Otherwise, the scale $\Delta$ is halved. Hence, every accepted update preserves rate feasibility and strictly reduces the CRB. The search terminates at a design whose feasible neighbors no longer improve the CRB at the finest scale $\Delta_{\min}$, or when the iteration budget $I_{\max}$ is exhausted. The stopping status records which of these two conditions ends the search.

	Repeating this search over a sequence of rate requirements $\{R_{\rm req}^{(k)}\}_{k=1}^{N_R}$ traces the C-R boundary. The two ends of this sequence are given by the communication and sensing endpoints $\boldsymbol z^{\rm F}$ and $\boldsymbol z^{\rm S}$. These endpoints are obtained by the same search without the rate constraint, maximizing $\widetilde R$ and minimizing $\widetilde C$, respectively, with $\alpha$ kept strictly below one. We process the requirements in descending order because each solution $\boldsymbol z^{\rm prev}$ satisfies $\widetilde R(\boldsymbol z^{\rm prev})\ge R_{\rm req}^{(k-1)}\ge R_{\rm req}^{(k)}$ and is thus a feasible warm start for the next requirement. To mitigate local minima, this warm start is complemented by the seed pool $\mathcal Z_0=\{\boldsymbol z^{\rm F},\boldsymbol z^{\rm S},\boldsymbol z^{\rm C}\}$, where $\boldsymbol z^{\rm C}$ is the box center. For the $k$-th requirement, the start set $\mathcal Z_k$ collects up to $N_{\rm start}$ well-separated feasible members of $\mathcal Z_0\cup\{\boldsymbol z^{\rm prev}\}$, ranked by their CRB values; if no member is feasible, $\mathcal Z_k$ is empty and the requirement is skipped. The results obtained from all starts in $\mathcal Z_k$ are collected in $\mathcal Q_k$. The performance archive $\mathcal P$ is initialized as $\mathcal P_0=\{(\widetilde R(\boldsymbol z^{\rm F}),\widetilde C(\boldsymbol z^{\rm F})),(\widetilde R(\boldsymbol z^{\rm S}),\widetilde C(\boldsymbol z^{\rm S}))\}$, and the nondominated subset of all resulting performance pairs, $\operatorname{ND}(\mathcal P)$, constitutes the estimated boundary. The overall procedure is summarized in Algorithm~\ref{alg:RCJO}.

	Let $C_{\rm phys}$ denote the cost of one steady-state evaluation. Since each iteration requires at most $|\mathcal D|=8$ new evaluations, the complexity of Algorithm~\ref{alg:RCJO} is $\mathcal O(N_RN_{\rm start}I_{\max}|\mathcal D|C_{\rm phys})$, excluding the fixed cost of computing the two seeds. Repeated evaluations are avoided by caching them in $\mathcal K$. 

	\section{Numerical Results}
	\label{sec:numerical}
	Building on the mechanism analysis in Section~\ref{sec:tradeoffs}, this section evaluates the achievable C-R tradeoff and the effects of sensing load and system parameters.

	\subsection{Simulation Setup}
	\label{sec:sim_setup}
	For the uplink communication link, a single antenna user in the far field is considered, and $\mathbf h_{\mathrm{com}}$ is the deterministic far-field channel with the angle-dependent ULA phase in \eqref{eq:field_com}. For the sensing task, the target is modeled as a point scatterer with unknown deterministic echo coefficients $|s_k|=1$, and the AoA CRB is evaluated over $K$ snapshots. The main simulation parameters are summarized in Table~\ref{tab:sim_params}.
	
	\begin{table}[t!]
		\caption{Simulation Parameters}
		\label{tab:sim_params}
		\centering
		\footnotesize
		\setlength{\tabcolsep}{6pt} 
		\renewcommand{\arraystretch}{1.4} 
		
		\begin{tabular}{@{} l | c | l | c @{}} 
			\hline
			\multicolumn{2}{c|}{Atomic Parameters} & \multicolumn{2}{c}{ISAC System Parameters} \\
			\hline
			\textbf{Parameter} & \textbf{Value} & \textbf{Parameter} & \textbf{Value} \\
			\hline
			Atomic species & $^{133}\text{Cs}$ & Comm. Freq. ($f_{\mathrm{LO},c}$) & $12.36$ GHz \\
			Probe laser ($\lambda_p$) & $852$ nm & Sens. Freq. ($f_{\mathrm{LO},s}$) & $17.17$ GHz \\
			Coupling laser ($\lambda_c$) & $510.5$ nm & Array Elements ($M$) & $8$ \\
			Input DC ($I_{\mathrm{in}}$) & $80~\mu$A & Element Spacing ($d$) & $8$ mm \\
			Probe Rabi ($\Omega_p$) & $0.8\Gamma_2$ & Snapshots ($K$) & $100$ \\
			Cell Length ($\ell_{\rm cell}$) & $2$ cm & Tx Power ($P_{\text{tx}}$) & $30$ dBm \\
			Beam Waist ($w_0$) & $500~\mu$m & Temperature ($T$) & $300$ K \\
			\hline
		\end{tabular}
		\\[2pt]
		\parbox{0.98\linewidth}{\footnotesize \emph{Note:} The dipole moments and decay rates are obtained from the ARC for $^{133}\mathrm{Cs}$ at $300$ K. The coupling Rabi frequency is searched over $\Omega_c/2\pi\in[0.4,15]$ MHz, corresponding approximately to $P_c\in[15,48]$ dBm under the adopted beam waist setting. Each operating point is evaluated from the full Liouvillian steady state and RF small signal linearization. The atomic number density is computed from the saturated Cs vapor density at the set temperature using \cite[Eq.~(1)]{steckCesiumDLineData}.}
	\end{table}   

	The proposed scheme, termed Joint Atomic Response MB-RAR, jointly optimizes $\Omega_c$ and $\alpha$ via Algorithm~\ref{alg:RCJO}. We compare it with the following baselines.
	\begin{itemize}
		\item \textbf{Communication Centric MB-RAR:} The coupling Rabi frequency is fixed at the operating point that maximizes the communication rate.
		\item \textbf{Sensing Centric MB-RAR:} The coupling Rabi frequency is fixed at the operating point that minimizes the sensing CRB.
		\item \textbf{RAR-ISAC (time sharing):} The receiver switches between the communication only and sensing only atomic operating endpoints.
		\item \textbf{Conventional RF-ISAC (time sharing):} A classical RF receiver baseline with time sharing \cite{hua2023mimo}.
	\end{itemize}
	
	\begin{figure}[!t]
		\centering
		\includegraphics[width=3.1in]{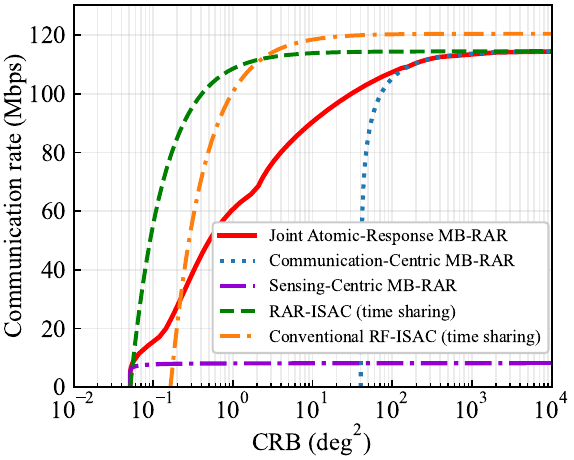}
		\caption{The CRB-rate region for the task of single target angle estimation.}
		\label{fig:Pareto_Front}
	\end{figure}
	
	\subsection{Characterization of the CRB-Rate Pareto Frontier}
	\label{sec:pareto_numerical}	
	
		Fig.~\ref{fig:Pareto_Front} depicts the CRB-rate regions achieved by the Joint Atomic Response MB-RAR, the Communication Centric and Sensing Centric MB-RAR, and the RAR-ISAC and Conventional RF-ISAC (time sharing) baselines under the default parameters. All boundaries exhibit the same monotonic trend: a higher communication rate is attained only at the cost of a larger sensing CRB. In MB-RAR, this tradeoff is dictated by the coupled variation of bandwidth, transconductance, and branch noise through the shared optical channel. The two ends of the figure reveal what limits each function. Toward the high-rate end, the achievable rate is bounded by the receive SNR rather than by the bandwidth. The conventional RF receiver attains the largest rate even though its fixed bandwidth is narrower than the RAR bandwidth near this endpoint, because its SNR is limited only by electrical noise. By contrast, the RAR SNR is further degraded by the finite RF-to-optical conversion and IBGI. For CRB, RAR offers a sensing advantage, whereas its rate is currently limited by conversion efficiency and quantum noise, which suggests the key direction for improving RAR-based ISAC.
		
		Between the two ends, the figure shows how effectively each scheme balances the two functions, which is governed by two factors: adaptability and concurrency. In terms of adaptability, the Joint Atomic Response MB-RAR encloses both fixed-coupling boundaries and achieves notably higher rates in the intermediate-CRB region. This gain arises because the coupling that best balances the two branches varies with the CRB requirement. Joint control retunes the common optical response accordingly, whereas each fixed setting is well matched only near its own endpoint. In terms of concurrency, RAR-ISAC (time sharing) exceeds simultaneous operation over part of the region. By dedicating each slot to a single-function operating point, time sharing removes the inter-branch gain competition present under concurrent operation. It pays for this through the time split, which lowers both the average rate and the accumulated Fisher information, and under the present parameters the former effect prevails within that part of the region. Where this gain vanishes, simultaneous operation is preferable. Moreover, simultaneous operation remains the only option when both functions must be continuously available, and the joint boundary quantifies its achievable performance.

\subsection{Communication Rate under Multi-Target Sensing Loads}
\label{sec:multi_target_load}

\begin{figure}[!t]
    \centering
    \includegraphics[width=\columnwidth]{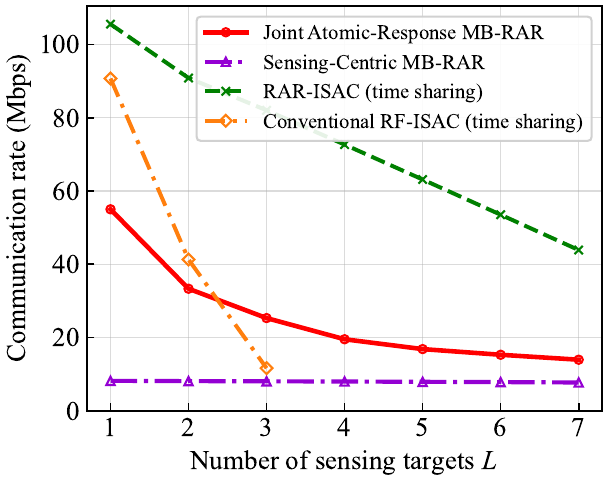}
    \caption{Communication rate versus the number of sensing targets under an average AoA CRB limit of $0.5~\mathrm{deg}^2$. Total echo power, array size, and snapshot count are fixed. Only feasible operating points are shown.}
    \label{fig:Multi_Target}
    \label{fig:multi_target_rate}
\end{figure}

	Fig.~\ref{fig:Multi_Target} evaluates the communication rate that each scheme can sustain as the sensing load grows, subject to the average AoA-CRB requirement $\overline C_L\le C_{\mathrm{lim}}$ with $C_{\mathrm{lim}}=0.5~\mathrm{deg}^2$. We consider $L=1,\ldots,7$ separable targets with $M=8$, $K=100$, $T=300$~K, and a fixed total echo power of $-118$~dBm, which is divided equally among the targets. The single target is placed at broadside, whereas for $L>1$ the target angles are uniformly spaced over $[-30^\circ,30^\circ]$. The remaining parameters follow Table~\ref{tab:sim_params}.
	
	Fig.~\ref{fig:Multi_Target} shows that the achievable rate of every feasible scheme decreases as the number of targets increases. Since the total echo power is fixed, each return carries only $1/L$ of it. Because the CRB is inversely proportional to the echo power, as indicated by \eqref{eq:CRB_sensing_final}, the Fisher information available for each target shrinks accordingly. Meeting the same average bound therefore requires a larger share of the receiver resources to be devoted to sensing, and this share is taken from communication. In this sense, the sensing load acts as an increasingly tight constraint that shifts the operating point along the CRB-rate boundary in Fig.~\ref{fig:Pareto_Front} toward the small-CRB end.
	
	The figure further reveals that the sensing load first determines feasibility before it determines rate. The Conventional RF-ISAC (time sharing) curve terminates once its pure-sensing bound $\overline C_{L,\mathrm S}$ exceeds $C_{\mathrm{lim}}$. Beyond this point, no time allocation can satisfy the requirement, since devoting the entire frame to sensing cannot outperform the all-sensing limit. Similarly, the Communication Centric MB-RAR yields no feasible point at any $L$, because its fixed coupling cannot provide sufficient sensing gain-to-noise ratio under this bound. This result highlights that the sensing advantage of RAR identified in Section~\ref{sec:pareto_numerical} directly extends the range of sensing loads that can be supported.
	
	Among the feasible schemes, the differences in both rate level and rate decay stem from how each scheme accommodates the additional sensing load. RAR-ISAC (time sharing) attains the highest rates in this comparison. Each interval operates at a pure-function atomic setting free of simultaneous IBGI, and the growing load is absorbed only through the time fraction, which leads to the linear decrease of $R$. The Joint Atomic Response MB-RAR absorbs the load by retuning the coupling strength together with the LO split, moving toward more sensitive atomic operating points at the cost of communication bandwidth and gain. By contrast, the Sensing Centric MB-RAR can adapt only the LO allocation. Since its fixed coupling already confines communication to a narrow response band, LO reallocation mainly changes the SNR within that band, which explains its lower rate and weaker dependence on $L$.

	\subsection{Performance Analysis of Communication and Sensing}
	
	Fig.~\ref{fig:parameter_sweeps} examines how the communication and sensing performance scale with the array size, RF power, and cell temperature. 
	Figs.~\ref{fig:parameter_sweeps}\subref{fig:sweep_rate_M} and~\ref{fig:parameter_sweeps}\subref{fig:sweep_crb_M} show that enlarging the array benefits both functions. As shown in Fig.~\ref{fig:parameter_sweeps}\subref{fig:sweep_rate_M}, receive combining raises the communication SNR, so the rate increases with $M$. More importantly, the MB-RAR curves overtake the Conventional RF-ISAC (time sharing) curve at larger arrays. In Fig.~\ref{fig:parameter_sweeps}\subref{fig:sweep_crb_M}, the CRB decreases with $M$ through the aperture factor $M(M^2-1)$ in \eqref{eq:CRB_sensing_final}.

	\begin{figure*}[!t]
		\centering
		\subfigure[Communication rate versus $M$.]
		{\includegraphics[width=0.32\textwidth]{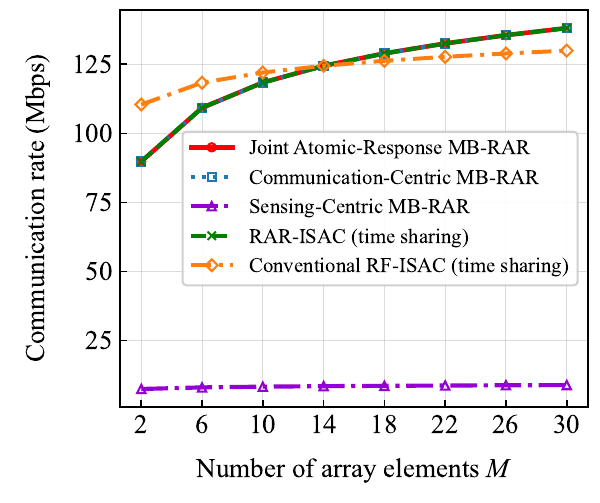}\label{fig:sweep_rate_M}}\hfill
		\subfigure[Communication rate versus $P_{\mathrm{tx}}$.]
		{\includegraphics[width=0.32\textwidth]{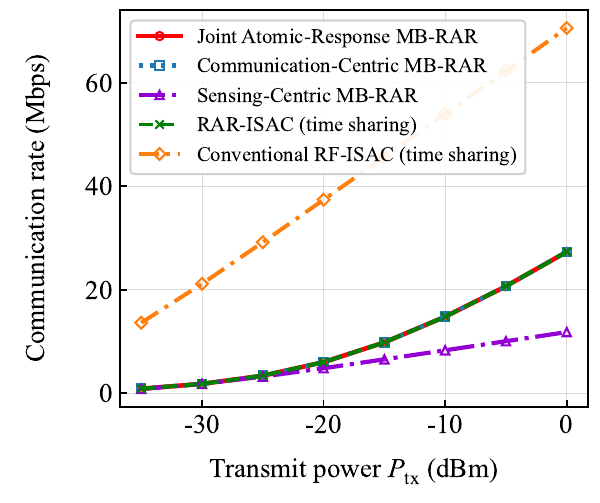}\label{fig:sweep_rate_power}}\hfill
		\subfigure[Communication rate versus temperature $T$.]
		{\includegraphics[width=0.32\textwidth]{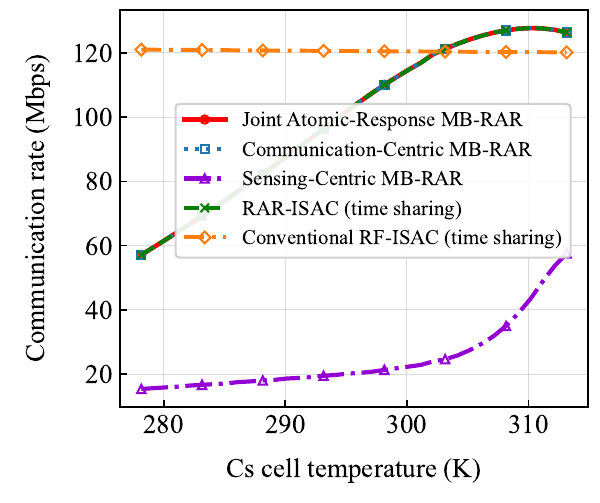}\label{fig:sweep_rate_temperature}}
		\par\vspace{0.5em}
		\subfigure[Sensing CRB versus $M$.]
		{\includegraphics[width=0.32\textwidth]{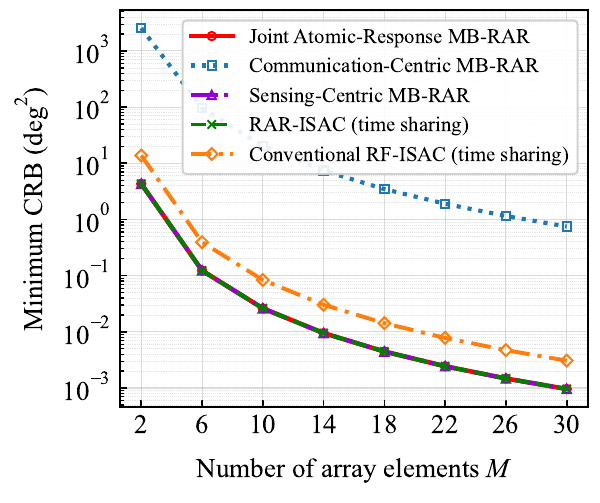}\label{fig:sweep_crb_M}}\hfill
		\subfigure[Sensing root-mean-square error (RMSE) versus $P_{\mathrm{echo}}$.]
		{\includegraphics[width=0.32\textwidth]{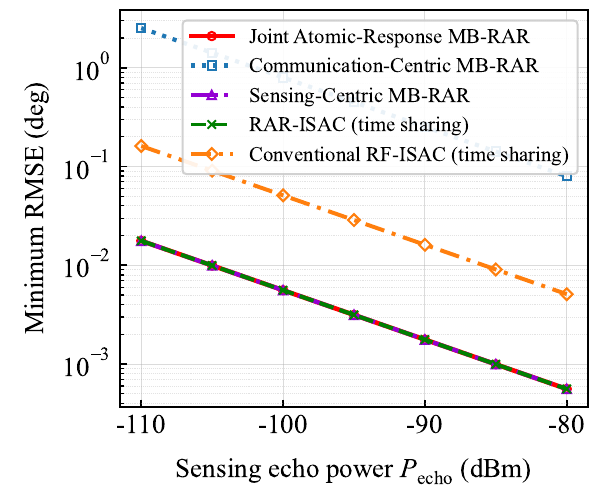}\label{fig:sweep_rmse_power}}\hfill
		\subfigure[Sensing CRB versus temperature $T$.]
		{\includegraphics[width=0.32\textwidth]{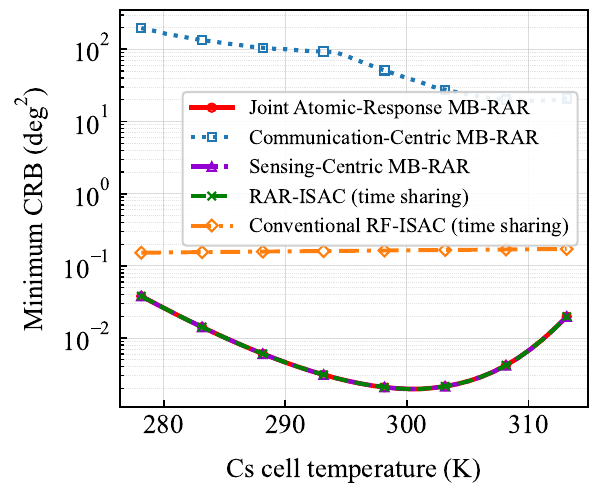}\label{fig:sweep_crb_temperature}}
		\caption{Communication and sensing performance under array-size, RF-power, and temperature sweeps.}
		\label{fig:parameter_sweeps}
		\label{fig:M_vs}
		\label{fig:TXP_vs}
		\label{fig:Temperature_vs}
	\end{figure*}
	
	Figs.~\ref{fig:parameter_sweeps}\subref{fig:sweep_rate_power} and~\ref{fig:parameter_sweeps}\subref{fig:sweep_rmse_power} sweep the communication transmit power $P_{\mathrm{tx}}$ and the incident echo power $P_{\mathrm{echo}}$, respectively. In Fig.~\ref{fig:parameter_sweeps}\subref{fig:sweep_rate_power}, increasing $P_{\mathrm{tx}}$ raises the received signal power without altering the noise at fixed atomic controls, and thus increases the rate of all schemes. The gap between the Joint Atomic Response and Communication Centric MB-RAR curves and the Sensing Centric MB-RAR curve widens with $P_{\mathrm{tx}}$. In Fig.~\ref{fig:parameter_sweeps}\subref{fig:sweep_rmse_power}, the RMSE curves, evaluated as $\mathrm{RMSE}=\sqrt{\mathrm{CRB}}$, decrease approximately in parallel on logarithmic axes, since \eqref{eq:CRB_sensing_final} indicates that the CRB is inversely proportional to $P_{\mathrm{echo}}$. The Joint Atomic Response and Sensing Centric MB-RAR improve this ratio, whereas the Communication Centric MB-RAR sacrifices it for communication bandwidth. Consequently, the sensing advantage of these two MB-RAR schemes remains approximately constant over the considered power range.

	Figs.~\ref{fig:parameter_sweeps}\subref{fig:sweep_rate_temperature} and~\ref{fig:parameter_sweeps}\subref{fig:sweep_crb_temperature} evaluate the influence of the cell temperature, which determines the atomic density. Following the saturated vapor-pressure model in \cite[Eq.~(1)]{steckCesiumDLineData}, the Cs vapor pressure is given by
	\begin{equation}
		\log_{10}\left(\frac{P_v(T)}{1~\mathrm{Torr}}\right)
		=
		\begin{cases}
			7.592-\dfrac{3999}{T}, & \text{(solid phase)},\\[4pt]
			7.046-\dfrac{3830}{T}, & \text{(liquid phase)},
		\end{cases}
		\label{eq:cs_vapor_pressure}
	\end{equation}
	where $P_v(T)$ is the saturated Cs vapor pressure. Using the ideal-gas relation, the corresponding saturated atomic number density is $N_{\mathrm{sat}}(T)={133.322~\mathrm{Pa}}/{k_BT}\left[{P_v(T)}/{1~\mathrm{Torr}}\right]~[\mathrm{m}^{-3}]$.

	Unlike the array size and RF power, the temperature affects RAR through two competing effects. On the one hand, a higher atomic density strengthens the optical susceptibility and hence the useful atomic response. On the other hand, the accompanying increase in optical depth attenuates the probe and reduces the transmitted photocurrent. In Fig.~\ref{fig:parameter_sweeps}\subref{fig:sweep_rate_temperature}, the Joint Atomic Response and Communication Centric MB-RAR curves rise over most of the temperature range, surpass the Conventional RF-ISAC (time sharing) curve at higher temperatures, and then turn slightly downward near the upper end. As the density grows, retuning the coupling field improves the attainable combination of bandwidth and receive SNR until probe attenuation begins to prevail. In Fig.~\ref{fig:parameter_sweeps}\subref{fig:sweep_crb_temperature}, the same competition produces an interior CRB minimum for the Joint Atomic Response and Sensing Centric MB-RAR. Below this minimum, the stronger atomic response improves the sensing gain-to-noise ratio. Beyond it, probe attenuation and the reoptimized conversion slope reduce the sensing transconductance by more than the accompanying reductions in shot noise and effective sensing bandwidth can compensate, so the CRB rises again. This balance is governed by the gain and bandwidth-dependent noise in \eqref{eq:CRB_sensing_final}, together with the temperature-dependent Doppler distribution and decay rates. The Conventional RF-ISAC (time sharing) baseline, which has no dependence on atomic density or optical conversion, varies only mildly and degrades slowly as the thermal-noise PSD increases.

	\section{Conclusion}
	\label{sec:conclusion}
	
In this paper, we investigated the fundamental C-R tradeoff of RAR-based ISAC. By solving the five-level Lindblad master equation at the LO-biased operating point and linearizing the RF signals as small perturbations, a Doppler-averaged double-sideband response model of the MB-RAR was established, which connects the microscopic atomic dynamics to the communication rate and the AoA CRB through the branch quantum transconductance. Based on this model, two intrinsic tradeoffs that have no counterpart in classical RF-ISAC were revealed. First, a stronger coupling laser broadens the dynamic response bandwidth but can weaken the atomic response. Second, strengthening one LO branch suppresses the conversion gain of the other branch through the shared Rydberg coherence, even when the latter's LO bias remains fixed, which was termed IBGI. Then, the achievable C-R region was formulated via the $\epsilon$-constraint method, and an RCJO algorithm was developed to characterize its Pareto boundary. Simulation results have demonstrated that joint atomic control yields a more complete C-R region than fixed-coupling operation, and that time sharing becomes preferable where inter-branch competition is strong. Compared with conventional RF-ISAC, the RAR achieves smaller sensing CRBs and remains feasible under heavier multi-target sensing loads. 

Future work will focus on the experimental validation of the MB-RAR response model, distributed RAR cooperation, reconfigurable microwave-transition selection, and physics-informed AI models for fast atomic-response calibration and optimization.

	\appendices
	\section{LO Sensitivities and Proof of Lemma~\ref{lem:cross_damping}}
	\label{app:proof_cross_damping}
	
	From $g_b=U_b\mathcal A_b$ and \eqref{eq:branch_gain_elasticities},
	\begin{equation}
		\partial_\alpha\ln g_b
		=\partial_\alpha\ln U_b
		+\sum_{b'\in\{\mathrm{com},\mathrm{sens}\}}
		\eta_{bb'}\partial_\alpha\ln U_{b'}.
		\label{eq:app_chain}
	\end{equation}
	Substitution of $\partial_\alpha\ln U_{\mathrm{com}}=1/(2\alpha)$ and $\partial_\alpha\ln U_{\mathrm{sens}}=-1/[2(1-\alpha)]$ yields \eqref{eq:communication_gain_competition} and \eqref{eq:cross_damping}. The following sensitivity equations determine the coefficients from the atomic response.

	Write $\Omega_{b',\mathrm{LO}}=U_{b'}e^{j\phi_{b'}}$, where $\phi_{b'}$ is a fixed LO phase. By \eqref{eq:branch_perturbation_superoperators}, the derivative of the unperturbed Liouvillian is
	\begin{equation}
		\frac{\partial\mathcal L_0(v)}{\partial U_{b'}}
		=e^{j\phi_{b'}}\mathcal L_{b'}^{\mathrm u}(v)
		+e^{-j\phi_{b'}}\mathcal L_{b'}^{\mathrm l}(v).
		\label{eq:liouvillian_lo_derivative}
	\end{equation}
	Consider a smooth operating point with a unique steady state and decaying trace-zero modes. Differentiating \eqref{eq:liouvillian_working_point} gives
	\begin{equation}
		\begin{aligned}
		\mathcal L_0(v)\frac{\partial\bar{\mathbf x}(v)}{\partial U_{b'}}
		&=-\frac{\partial\mathcal L_0(v)}{\partial U_{b'}}\bar{\mathbf x}(v),\\
		\mathrm{Tr}\!\left(\frac{\partial\bar{\boldsymbol\rho}(v)}{\partial U_{b'}}\right)=\mathbf t^{\mathrm T}\frac{\partial\bar{\mathbf x}(v)}{\partial U_{b'}}&=0.
		\end{aligned}
		\label{eq:steady_state_lo_sensitivity}
	\end{equation}
	This trace-constrained system determines the change in the atomic operating point at each velocity.

	For $\nu\in\{\mathrm u,\mathrm l\}$, let $\mathbf y_b^\nu(\omega;v)$ denote the common-state response to a unit excitation of source channel $\nu$ on branch $b$. The channel labels follow \eqref{eq:branch_dynamic_response}; both LO biases remain in the common working point. It satisfies
	\begin{equation}
		\begin{aligned}
		(j\omega\mathbf I-\mathcal L_0(v))\mathbf y_b^\nu
		&=\mathcal L_b^\nu(v)\bar{\mathbf x}(v),\\
		\mathbf t^{\mathrm T}\mathbf y_b^\nu&=0,
		\end{aligned}
		\label{eq:state_response_trace_system}
	\end{equation}
	with $F_b^\nu(\omega;v)=\mathbf c_{21}^{\mathrm T}\mathbf y_b^\nu(\omega;v)$. Differentiation at fixed $\omega$ yields
	\begin{equation}
		\begin{aligned}
		(j\omega\mathbf I-\mathcal L_0(v))
		\frac{\partial\mathbf y_b^\nu}{\partial U_{b'}}
		&=\frac{\partial\mathcal L_0(v)}{\partial U_{b'}}\mathbf y_b^\nu +\mathcal L_b^\nu(v)
		\frac{\partial\bar{\mathbf x}(v)}{\partial U_{b'}},\\
		\mathbf t^{\mathrm T}\frac{\partial\mathbf y_b^\nu}{\partial U_{b'}}&=0.
		\end{aligned}
		\label{eq:dynamic_response_lo_sensitivity}
	\end{equation}
	The two terms on the right-hand side describe the changes in the evolution operator and the steady-state source, respectively. The trace constraint also specifies the solution at $\omega=0$.

	At fixed temperature, differentiating \eqref{eq:branch_sideband_response_doppler} and applying \eqref{eq:branch_iq_response} gives
	\begin{equation}
		\frac{\partial G_{b,\mathrm I}^{\mathrm D}(\omega)}{\partial U_{b'}}
		=\int_{-\infty}^{\infty}\mathbf c_{21}^{\mathrm T}
		\left(\frac{\partial\mathbf y_b^{\mathrm u}}{\partial U_{b'}}
		+\frac{\partial\mathbf y_b^{\mathrm l}}{\partial U_{b'}}\right)
		f_T(v)\,\mathrm dv.
		\label{eq:doppler_response_lo_sensitivity}
	\end{equation}
	The steady-state susceptibility derivative follows from \eqref{eq:chi_dynamic}:
	\begin{equation}
		\frac{\partial\bar\chi^{\mathrm D}}{\partial U_{b'}}
		=-\frac{2N|\mu_{21}|^2}{\epsilon_0\hbar\Omega_p}
		\int_{-\infty}^{\infty}\mathbf c_{21}^{\mathrm T}
		\frac{\partial\bar{\mathbf x}(v)}{\partial U_{b'}}f_T(v)\,\mathrm dv.
		\label{eq:susceptibility_lo_sensitivity}
	\end{equation}
	Define $\mathcal G_b\triangleq\mathcal M\{G_{b,\mathrm I}^{\mathrm D}\}(\omega_b)$, the projected branch response at its IF. At fixed IF, its derivative is obtained by applying the same projection, including the negative-frequency conjugate, to \eqref{eq:doppler_response_lo_sensitivity}. Differentiating \eqref{eq:IDC_liouvillian} and \eqref{eq:g_def}, and using $\mathcal A_b=g_b/U_b$, yields
	\begin{equation}
		\begin{aligned}
		\eta_{bb'}
		&=-U_{b'}k_p\ell_{\mathrm{cell}}
		\operatorname{Im}\!\left\{\frac{\partial\bar\chi^{\mathrm D}}{\partial U_{b'}}\right\}\\
		&\quad+\frac{U_{b'}}{|\mathcal G_b|^2}
		\operatorname{Re}\!\left\{\mathcal G_b^*
		\frac{\partial\mathcal G_b}{\partial U_{b'}}\right\}
		-\delta_{bb'},
		\end{aligned}
		\label{eq:gain_elasticity_liouvillian}
	\end{equation}
	where $\delta_{bb'}$ is one for $b=b'$ and zero otherwise. Equations~\eqref{eq:steady_state_lo_sensitivity}--\eqref{eq:gain_elasticity_liouvillian} complete the evaluation of the branch-gain derivatives in Lemma~\ref{lem:cross_damping}.

	\section{Proof of Lemma~\ref{lem:CRB_AoA}}
	\label{app:proof_CRB_AoA}
	For brevity, define the effective channel gain as $h_{\mathrm{sens}}=\sqrt{\eta_0P_{\mathrm{echo}}R_{\mathrm{L}}/A_{\mathrm{e}}}\,\ell_{\rm cell}g_{\mathrm{sens}}$ and let $\boldsymbol{\mu}_k(\theta,s_k)=h_{\mathrm{sens}}s_k\mathbf{a}(\theta)$. Then, the $K$ snapshots are independent with $\mathbf{y}_{\mathrm{sens},k}\sim \mathcal{CN}(\boldsymbol{\mu}_k, \sigma_{n,\mathrm{sens}}^2\mathbf{I}_M)$. We define the unknown parameter vector as $\boldsymbol{\xi}=[s_{R,1}, s_{I,1},\ldots,s_{R,K}, s_{I,K}, \theta]^{\mathrm T}$, with $s_k = s_{R,k} + j s_{I,k}$. The nonzero partial derivatives are
	\begin{equation}
		\frac{\partial \boldsymbol{\mu}_k}{\partial s_{R,k}} \!=\! h_{\mathrm{sens}}\mathbf{a}(\theta),
		\frac{\partial \boldsymbol{\mu}_k}{\partial s_{I,k}} \!=\! j h_{\mathrm{sens}}\mathbf{a}(\theta),
		\frac{\partial \boldsymbol{\mu}_k}{\partial \theta} \!=\! h_{\mathrm{sens}} s_k\dot{\mathbf{a}}(\theta).
	\label{eq:crb_mean_derivatives}
\end{equation}
	The Fisher information matrix (FIM) can be partitioned as $\mathbf{J} \!\!=\!\! \begin{bmatrix} \mathbf{J}_{ss} & \mathbf{J}_{s\theta}\\ \mathbf{J}_{\theta s} & J_{\theta\theta} \end{bmatrix}$. Since $\mathbf J_{ss}$ is block diagonal over the snapshots, eliminating the $2K$ nuisance parameters via the Schur complement yields the equivalent FIM for $\theta$
	\begin{equation}
		\begin{aligned}
		\mathcal{J}(\theta) &= J_{\theta\theta} - \mathbf{J}_{\theta s}\mathbf{J}_{ss}^{-1}\mathbf{J}_{s\theta}\\
		&= \frac{2|h_{\mathrm{sens}}|^2}{\sigma_{n,\mathrm{sens}}^2}\sum_{k=1}^{K}|s_k|^2 \|\mathbf{P}_{\perp}(\theta)\dot{\mathbf{a}}(\theta)\|^2\\
		&=\frac{2K|h_{\mathrm{sens}}|^2}{\sigma_{n,\mathrm{sens}}^2} \|\mathbf{P}_{\perp}(\theta)\dot{\mathbf{a}}(\theta)\|^2,
		\end{aligned}
	\label{eq:crb_schur_information}
\end{equation}
	where the last equality uses $|s_k|=1$, and
	$\mathbf{P}_{\perp}(\theta)$ is the orthogonal projector onto the null space of $\mathbf{a}(\theta)$. For a ULA, shifting the element index to $m=0, \dots, M-1$ without loss of generality, the phase term is $\psi=kd\sin\theta$, and $a_m(\theta)=e^{jm\psi}$. Consequently, the derivative is $\dot{a}_m(\theta) = j m (kd\cos\theta)e^{jm\psi}$. Substituting this derivative and $\|\mathbf a(\theta)\|^2=M$ into $\|\mathbf P_{\perp}\dot{\mathbf a}\|^2=\|\dot{\mathbf a}\|^2-|\mathbf a^{\mathrm H}\dot{\mathbf a}|^2/M$ gives
	\begin{equation}
		\|\mathbf{P}_{\perp}(\theta)\dot{\mathbf{a}}(\theta)\|^2 = (kd\cos\theta)^2\frac{M(M^2-1)}{12}.
	\end{equation}
	Substituting this result into $\mathrm{CRB}(\theta) = \mathcal{J}^{-1}(\theta)$ completes the proof of \eqref{eq:CRB_sensing_final}.
	
	\bibliographystyle{IEEEtran}
	\bibliography{Ref}
\end{document}